\documentclass[11pt]{article}
\usepackage[utf8]{inputenc}
\usepackage{amsmath,amsthm,amsfonts,fullpage, amssymb, xcolor,hyperref,bbm,graphicx}
\usepackage[ruled,noend,noline]{algorithm2e}
\usepackage{authblk}
\usepackage{mathtools}
\usepackage{cleveref}
\usepackage{natbib}
\usepackage{macros}

\usepackage{amsthm}
\usepackage{aliascnt}
\usepackage{cleveref}

\newtheorem{theorem}{Theorem} 

\newcommand{\makenamedtheorem}[2]{%
  \newaliascnt{#1}{theorem}
  \newtheorem{#1}[#1]{#2}
  \aliascntresetthe{#1}
}

\makenamedtheorem{corollary}{Corollary}
\makenamedtheorem{conjecture}{Conjecture}
\makenamedtheorem{lemma}{Lemma}
\makenamedtheorem{proposition}{Proposition}
\makenamedtheorem{protocol}{Protocol}
\makenamedtheorem{claim}{Claim}
\makenamedtheorem{fact}{Fact}
\makenamedtheorem{assumption}{Assumption}

\theoremstyle{definition}
\makenamedtheorem{example}{Example}
\makenamedtheorem{problem}{Problem}
\makenamedtheorem{definition}{Definition}
\makenamedtheorem{intuition}{Intuition}
\makenamedtheorem{idea}{Idea}
\makenamedtheorem{exercise}{Exercise}
\makenamedtheorem{remark}{Remark}

\crefname{theorem}{Theorem}{Theorems}
\Crefname{theorem}{Theorem}{Theorems}

\crefname{lemma}{Lemma}{Lemmas}
\Crefname{lemma}{Lemma}{Lemmas}

\crefname{corollary}{Corollary}{Corollaries}
\Crefname{corollary}{Corollary}{Corollaries}

\crefname{conjecture}{Conjecture}{Conjectures}
\Crefname{conjecture}{Conjecture}{Conjectures}

\crefname{proposition}{Proposition}{Propositions}
\Crefname{proposition}{Proposition}{Propositions}

\crefname{protocol}{Protocol}{Protocols}
\Crefname{protocol}{Protocol}{Protocols}

\crefname{claim}{Claim}{Claims}
\Crefname{claim}{Claim}{Claims}

\crefname{fact}{Fact}{Facts}
\Crefname{fact}{Fact}{Facts}

\crefname{assumption}{Assumption}{Assumptions}
\Crefname{assumption}{Assumption}{Assumptions}

\crefname{example}{Example}{Examples}
\Crefname{example}{Example}{Examples}

\crefname{problem}{Problem}{Problems}
\Crefname{problem}{Problem}{Problems}

\crefname{definition}{Definition}{Definitions}
\Crefname{definition}{Definition}{Definitions}

\crefname{intuition}{Intuition}{Intuitions}
\Crefname{intuition}{Intuition}{Intuitions}

\crefname{idea}{Idea}{Ideas}
\Crefname{idea}{Idea}{Ideas}

\crefname{exercise}{Exercise}{Exercises}
\Crefname{exercise}{Exercise}{Exercises}

\crefname{remark}{Remark}{Remarks}
\Crefname{remark}{Remark}{Remarks}

\newif\ifarxiv

\title{New Bounds for Circular Trace Reconstruction}
\begin{document}

\arxivtrue

\author{
Arnav Burudgunte \quad Paul Valiant\thanks{Partially supported by NSF award CCF-2127806 and by Office of Naval Research award N000142412695.} \quad Hongao Wang \\
Purdue University \\
}

\maketitle
\begin{abstract}

The ``trace reconstruction'' problem asks, given an unknown binary string $x$ and a channel that repeatedly returns ``traces'' of $x$ with each bit randomly deleted with some probability $p$, how many traces are needed to recover $x$? There is an exponential gap between the best known upper and lower bounds for this problem. Many variants of the model have been introduced in hopes of motivating or revealing new approaches to narrow this gap. We study the variant of \emph{circular} trace reconstruction introduced by Narayanan and Ren (ITCS 2021), in which traces undergo a random cyclic shift in addition to random deletions.

We show an improved lower bound of $\tilde{\Omega}(n^5)$ for circular trace reconstruction. This contrasts with the (previously) best known lower bounds of $\tilde{\Omega}(n^3)$ in the circular case and $\tilde{\Omega}(n^{3/2})$ in the linear case. Our bound shows the indistinguishability of traces from two \emph{sparse} strings $x,y$ that each have a \emph{constant} number of nonzeros. Can this technique be extended significantly? How hard is it to reconstruct a \emph{sparse} string $x$ under a cyclic deletion channel? We resolve these questions by showing, using Fourier techniques, that $\tilde{O}(n^6)$ traces suffice for reconstructing any constant-sparse string in a circular deletion channel, in contrast to the upper bound of $\exp(\tilde{O}(n^{1/3}))$ for general strings in the circular deletion channel. This shows that new algorithms or new lower bounds must focus on \emph{non-constant-sparse} strings.

\end{abstract}

\section{Introduction}

Let $x$ be a binary string of length $n$. Given sample access to a deletion channel 
which deletes each character in $x$ with probability $p$ and returns the remaining subsequence, the trace reconstruction problem asks how many samples are necessary to recover $x$ reliably. This problem has been studied extensively in various forms and under various assumptions, such as \cite{Batu2004reconstruct, ban2019beyond, chase2021new, chase2021separating, krishnamurthy2021trace, rivkin25a}. In the standard setting, where the deletion probability $p$ is a constant and the original string $x$ is arbitrary, there remains an exponential gap between the best known upper bound $\exp(\tilde{O}(n^{1/5}))$ and the best known lower bound $\tilde{\Omega}(n^{3/2}))$ for worst-case reconstruction \cite{chase2021new, chase2021separating}. In hopes of narrowing this gap, many related problems have been proposed and studied, such as matrix trace reconstruction \cite{krishnamurthy2021trace}, population recovery \cite{ban2019beyond}, and coded trace reconstruction \cite{cheraghchi2020coded, brakensiek2020coded}. 

We study one such mild variant, the problem of \textit{circular trace reconstruction} introduced by \cite{narayanan2021circular}. In this problem, the original string $x$ undergoes a random cyclic shift after deletion, and the reconstruction algorithm must return any string which is cyclically equivalent to $x$. Circular trace reconstruction is known to be at least as hard as the linear version, with the best known bounds in each model being fairly similar: prior to our work, the best known lower bound was $\tilde{\Omega}(n^3)$ in the circular model compared to $\tilde{\Omega}(n^{3/2})$ in the linear model. The best known upper bound for arbitrary strings is $\exp(\tilde{O}({n^{1/3}}))$ in the circular model (assuming $n$ can be written as a product of two or fewer primes), compared to $\exp(\tilde{O}({n^{1/5}}))$ in the linear model. 

\subsection{Our Contributions}


We show that, in the worst case, $\tilde{\Omega}(n^5)$ traces are necessary to distinguish two 
strings, implying a $\tilde{\Omega}(n^5)$ lower bound for the general problem of circular trace reconstruction. This improves the previous $\tilde{\Omega}(n^3)$ proved by \cite{narayanan2021circular}. 

We prove this result by exhibiting two cyclically distinct strings $x,y$ that have a \emph{constant} number of nonzero entries, yet which need $\tilde{\Omega}(n^5)$ traces to distinguish. This strong lower bound motivates the general question: how difficult is it to reconstruct constant-sparse strings $x$ from circular traces? To what degree can this lower bound of $\tilde{\Omega}(n^5)$ be improved?

We resolve these questions almost-tightly by showing that $\tilde{O}(n^6)$ traces suffice to distinguish any two sparse circular strings. As no polynomial-sample algorithm is known for the general case, our results imply that stronger lower bounds can only be obtained from non-constant-sparse strings, and also that new algorithmic improvements should focus on the non-constant-sparse case. 

\ifarxiv
\medskip\noindent{\bf Techniques:} Both our upper bound and lower bound rely on an analysis of a new cyclically invariant quantity that we call \emph{cyclic statistics}. 
For a binary string $x$ with $k$ nonzero entries, $1\underbrace{0\ldots 0}_{x_1}1\underbrace{0\ldots 0}_{x_2}\ldots1\underbrace{0\ldots 0}_{x_k}$, we denote it here as $10^{x_1}\ldots10^{x_k}$. For a positive integer $m$, given an $m$-tuple $(i_1,\ldots,i_m)$ of integers mod $k$, we define the associated $m^\textrm{th}$ order cyclic statistic of $x$ as a degree $m$ polynomial in $x_1,\ldots,x_k$, produced by taking a monomial function of the $x_j$s and summing the function over all cyclic shifts of $x_1,\ldots,x_k$.

Our lower and upper bounds come from the following two results we show about cyclic statistics. First, there are cyclically distinct sparse strings $x,y$ that have \emph{identical} cyclic statistics up to order 4. By contrast, surprisingly, every pair of cyclically distinct strings must differ in some cyclic statistic of order $m\leq 6$.

The first result leads to our lower bound, when combined with a proof that: if two sparse strings $x$ and $y$ have identical cyclic statistics up to order $m$, the distributions of the traces generated by $x$ and $y$ will be the same up to $\frac{\log^{m+1}n}{n^{m+1}}$ in Hellinger distance.
The second result leads to our upper bound, by showing how to distinguish $m^\textrm{th}$ order cyclic statistics using $\tilde{O}(n^m)$ traces.
As the cyclic statistics are well-connected with both lower bound and upper bound in the sparse circular strings reconstruction problem, these cyclic statistics analysis tools may be of independent interest as well.
\fi

\subsection{Technical Overview}

We focus on the specific case of sparse circular strings, defined as binary strings which contain $k$ $1$s for some constant $k$. This family of strings has the useful property that any $k$-sparse string $x$ is a cyclic shift of $10^{x_1}10^{x_2}\ldots10^{x_k}$; therefore, $x$ can be represented as the integer sequence $(x_1,\ldots,x_k)$, where $x_j$ specifies the number of $0$s (the ``gap'') between the $j$th and $(j+1)$th $1$s in $x$. With probability $(1-p)^k$ the deletion channel will preserve all $1$s in $x$, yielding a trace $\tx$ which can also be represented by a sequence of $k$ gaps. Crucially, in our setting of constant $p,k$, this probability is constant. When this occurs, the distribution over traces $\tx$ depends only on the original integer sequence of gaps in $x$. Both our upper and lower bounds are based on properties of this integer sequence.

\ifarxiv
    \paragraph{Cyclic Statistics.}
\else
    \subparagraph*{Cyclic Statistics.} 
\fi
Our main technical tool is a family of shift-invariant functions which we call \textit{cyclic statistics}. At a high level, a cyclic statistic of a sequence $x = (x_1,\ldots,x_k)$ is a cyclically-invariant polynomial of the variables $x_1,\ldots,x_k$; the statistic is obtained by summing a monomial function over all cyclic shifts of $x$. The \textit{order} of this statistic is the degree of the corresponding monomial. If two strings $x$ and $y$ are cyclic shifts of each other, their cyclic statistics are clearly identical. Surprisingly, we prove that the converse holds for degree 6: if $x$ and $y$ are cyclically distinct, there exists a cyclic statistic of degree $\leq 6$ in which they differ. We give a lower bound showing that this result is nearly tight. 

\begin{lemma}[Informal]
    \label{lem:stat-different}
    Let $x$ and $y$ be two cyclically distinct integer sequences. Then $x$ and $y$ differ in some cyclic statistic of order at most $6$. Conversely, there exist integer sequences $x$ and $y$ with identical cyclic statistics up to order $4$.
\end{lemma}

We prove the upper bound in \Cref{lem:stat-different} using Fourier techniques (see Lemma~\ref{lem:cyclic-stats} in Section~\ref{sec:cyclic}). We show that if $x$ and $y$ have identical cyclic statistics up to order $6$, their Fourier transforms $\hat{x}$ and $\hat{y}$ must satisfy a certain system of sparse linear equations, each involving $6$ or fewer variables. We then analyze this system and show that any solution $\hat{x},\hat{y}$ has the property that there exists an integer $c$ for which \begin{equation}\label{eq:cyclic-conclusion}\hat{x}_j = \hat{y}_j \exp\left(2\pi i \cdot \frac{cj}{k}\right)\quad \textrm{for all } j \in [k].\end{equation} By basic properties of Fourier transforms, this implies that $y$ is a cyclic shift of $x$. 

The significance of the order $6$ stems from a number theoretic result. Our original constraints involve an equation of two variables for each $j \in [k]$, but our analysis only applies for $j$ relatively prime to $k$. We show that each $j \in [k]$ can be written as the sum of 2 or 3 integers relatively prime to $k$, transforming our constraint into a 6-variable equation, which implies Equation~\ref{eq:cyclic-conclusion}, given all cyclic statistics up to order 6 vanish.

The lower bound---that there exist $x$ and $y$ with identical cyclic statistics up to order 4---is proved directly. We provide a pair of sequences of length 12 satisfying this property, obtained via computer search: $x_{(j)}=(0,2,3,2,1,1,1,1,2,3,2,0)$ and the sequence $y_{(j)}=3-x_{(j)}$.  

Cyclic statistics are preserved in expectation by the cyclic deletion channel, leading to our algorithmic results. They also characterize the distribution of traces obtained from the deletion channel, from which we obtain our information-theoretic lower bound. We give a brief overview of this connection and our results below.

\ifarxiv
    \paragraph{Lower Bound}
\else
    \subparagraph*{Lower Bound.} 
\fi
In order to build our lower bound, we relate the probability distribution over traces of $x$ to the cyclic statistics of $(x_1,\ldots,x_k)$. Roughly, we show that the probability of generating a trace $\tx$ from $x$ is proportional to a polynomial function of $n$ and that the coefficients of this polynomial depend only on the cyclic statistics of $x$. We use this fact to show that if two cyclically distinct strings $x$ and $y$ have identical cyclic statistics up to order $m$, the Hellinger distance between the distributions of the traces generated from $x$ and $y$ is $\tilde{O}(n^{-m-1})$. Therefore, it requires $\tilde{\Omega}(n^{m+1})$ samples to distinguish the two distributions.

\begin{theorem}[Informal version of \Cref{lem:lowerbd}]
    \label{lem:lower-bound-informal}
    Let $(x_1,\ldots,x_k)$ and $(y_1,\ldots,y_k)$ be two sequences that are permutations of each other and have identical cyclic statistics up to order $m$. Then for any constant deletion probability $p$, any algorithm which distinguishes the strings
    \begin{equation*}
        x = 10^{n+x_1}10^{n+x_2}\ldots10^{n+x_k} \text{ and } y = 10^{n+y_1}10^{n+y_2}\ldots10^{n+y_k}
    \end{equation*} 
    requires $\tilde{\Omega}(n^{m+1})$ traces from a cyclic deletion channel. 
\end{theorem}

By \Cref{lem:stat-different}, there exist two cyclically distinct sequences $x$ and $y$ with identical cyclic statistics up to order $4$. Plugging the corresponding strings into \Cref{lem:lower-bound-informal}, we obtain our $\tilde{\Omega}(n^5)$ lower bound for circular trace reconstruction. 

\ifarxiv    
    \paragraph{Upper Bound}
\else
    \subparagraph*{Upper Bound.} 
\fi
We build an algorithm that can recover any $k$-sparse string $x$ using $\tilde{O}(n^6)$ samples from the deletion channel, showing that our techniques from above \emph{cannot} be extended to yield a significantly stronger lower bound. 

\begin{theorem}[Informal version of \Cref{thm:upperbd}]
    Let $x$ and $y$ be two cyclically distinct constant-sparse binary strings of length $n$. Then there exists an algorithm which, for any constant deletion probability $p$, distinguishes $x$ from $y$ using $\tilde{O}(n^6)$ traces from a cyclic deletion channel.
\end{theorem}

As explained above, our algorithm exploits the fact that the retention probability $q \coloneqq 1-p$ is a constant. Therefore, with constant probability, no 1 is deleted in $x$. Assuming this occurs, every trace $\tx$ can be expressed as a sequence of binomial random variables; that is, $\tilde{x} = (\tilde{x}_1,\ldots,\tilde{x}_k)$ with $\tilde{x}_j \sim \Bin{x_j}{q}$. By \Cref{lem:stat-different}, any two cyclically distinct sequences $x$ and $y$ differ in some cyclic statistic of order at most $6$. As all cyclic statistics of integer sequences are integers, the difference between $x$ and $y$ in this statistic is at least $1$. Thus we aim to estimate cyclic statistics from traces to within constant error. 

A na\"ive estimator is the empirical average of a certain cyclic statistic of the traces. However, this na\"ive algorithm requires $\Omega(n^{11})$ samples.
This is because, in the worst case, each $\tx_j$ may have mean  and variance both $\Omega(n)$. Thus the variance of the product of $6$ such random variables may be $\Omega(n^{11})$. However, instead, if we can shift the binomials such that they become random variables whose absolute value is $\tilde{O}(\sqrt{n})$, the absolute value of their product will be bounded by $\tilde{O}(n^3)$. 
We accomplish this by partitioning the set of expectations $\{qx_j\}$ into constantly many clusters. All elements in the same cluster are $\tilde{O}(\sqrt{n})$-close to each other and $\tilde{\Omega}(\sqrt{n})$-far from all other clusters. Since each binomial is tightly concentrated around its expectation, we can map each $\tx_j$ to the correct cluster with high probability, and subtract the cluster center from $\tx_j$ to obtain a random variable whose absolute value is bounded by $\tilde{O}(\sqrt{n})$ with high probability. 


A final hurdle is that subtracting cluster centers may erase any differences in the cyclic statistics of $x$ and $y$. See the discussion at the start of Section~\ref{sec:upperbound} and Example~\ref{ex:cyclic} for more algorithmic intuition, and discussion of why \textit{cyclic statistics mod $\ell$} are needed to yield a $\tilde{O}(n^6)$ algorithm. See \Cref{alg:tester} and \Cref{thm:upperbd} for our main algorithmic result.


\subsection{Structure of the Rest of the Paper}
In Section~\ref{sec:related} we discuss related work. Section~\ref{sec:defs} introduces the key definitions used in this paper. The rest of the paper consists of three technical sections that may be read independently, in any order. Section~\ref{sec:cyclic} builds tools to analyze cyclic statistics, leading to our characterization that ``any two cyclically distinct strings $x,y$ must differ in \emph{some} cyclic statistic of order $m\leq 6$.'' Section~\ref{sec:upperbound} contains our algorithmic results, using the characterization of Section~\ref{sec:cyclic} to give an algorithm that distinguishes any constant-sparse strings with $\tilde{O}(n^6)$ traces. Section~\ref{sec:lower} shows our lower bounds, based on a Hellinger distance analysis, showing that two strings $x,y$ that are permutations of each other and have identical cyclic statistics up to order $m$ require $\tilde{\Omega}(n^{m+1})$ traces to distinguish.

\subsection{Related Work}\label{sec:related}

There is a substantial amount of work related to the trace reconstruction problem and its variants. 

The problem was introduced in \cite{levenshtein2002efficient,Batu2004reconstruct}.
\cite{Batu2004reconstruct} discusses the original linear version of the trace reconstruction problem with deletion probability $O(1/\log n)$ and provides a simple and intuitive algorithm to reconstruct the strings. 
\cite{holenstein08trace} gives an algorithm for the trace reconstruction problem with constant deletion probability using $\e(\tilde{O}(n^{1/2}))$ samples. 
This result was improved independently by \cite{de2019optimal} and \cite{nazarov16Trace}, which present an algorithm using $\e(O(n^{1/3}))$ samples. 
More recently, \cite{chase2021separating} provides an algorithm that can reconstruct any string after a deletion channel with constant deletion probability with $\e(\tilde{O}(n^{1/5}))$ traces. 

On the lower bound side, \cite{Batu2004reconstruct} provides a $\Omega(n)$ lower bound, showing that distinguishing the two strings $0^n10^{n-1}$ and $0^{n-1}10^n$ requires $\Omega(n)$ samples. 
Then \cite{holden20Lower} and \cite{chase2021new} improved the lower bound by replacing the $0$'s in the previous strings by alternating $01010101\ldots$. 
However, the best known lower bound, from the work of \cite{chase2021new} is still $\tilde{\Omega}(n^{3/2})$, and there is still an exponential gap between the best known upper bound and lower bound. 

Restricting the class of algorithms to yield improved lower bounds, \cite{de2019optimal} shows that all ``mean-based'' algorithms require $\e(\Omega(n^{1/3}))$ samples to reconstruct strings. 
Here, a mean-based algorithm is any algorithm that only uses the empirical means of individual bits of the traces. 
Extending this line of work, \cite{kuan24on} considers ``$k$-mer-based'' algorithms which rely on counting the occurrences of contiguous $k$-bit strings (where ``mean-based'' algorithm are the case $k=1$). \cite{kuan24on} shows that, even for  $k$ up to $n^{1/5}$, any $k$-mer-based algorithm requires $\e(\Omega(n^{1/5}))$ samples to reconstruct the string. 
All known algorithms for the worst-case trace reconstruction problem are $k$-mer-based, including the algorithm presented by \cite{chase2021separating}. More recently, \cite{chen24trace} strengthened this result to show an analogous $\e(\Omega(n^{1/5}))$-sample lower bound for all statistical query (SQ) algorithms whose queries are ``$\tilde{O}(n^{1/5})$-local.'' Thus if there is an algorithm using fewer than $\e(n^{1/5})$ samples, it must use ``non-local information.'' We point out that the algorithm of the present paper, though we are in the rather different setting of \emph{sparse} and \emph{circular} trace reconstruction, looks at more global properties of the input (see Algorithm~\ref{alg:tester} and Theorem~\ref{thm:upperbd}).

As the general worst-case trace reconstruction problem still has an exponential gap, much work has developed and investigated variants of the problem. 
\cite{krishnamurthy2021trace} investigates many different variants, including matrix reconstruction and trace reconstruction for sparse strings. 
Most relevant to our work: in the regime where $k$ and $p$ are both constants, they analyze the ``trivial'' algorithm of directly estimating the size of each gap, giving an $O(n\log n)$ upper bound. (Directly estimating the size of each gap does not work in our circular trace reconstruction setting, since each trace gets randomly cyclically shifted, and thus a gap in a trace could have come from \emph{any} of the $k$ gaps in the original string, and thus we cannot aggregate statistics about any individual gap.) 
\cite{aamand2025near} provides an algorithm using $O(n\log n)$ samples for the trace reconstruction problem for ``separated'' strings, which means the number of $0$'s between two $1$'s is at least $\text{polylog } n$. 
The work of \cite{grigorescu_limitations_2022} and \cite{sima_trace_2021} shows that when the string is confined in a known Hamming ball or an edit distance ball of radius $k$, there are algorithms that can recover the string with $n^{O(k)}$ traces.
\cite{peres17Average} investigates the average-case trace reconstruction problem, that is, if the string is uniformly randomly chosen instead of adversarially generated, how many traces are required to recover it with high probability. 
They show that, in contrast to the worst-case, it only requires $\e(O(\log^{1/2}n))$ traces in the average case. 
\cite{cheraghchi2020coded} investigates the coded version of the trace reconstruction problem, that is, if we code the original string with a high-rate error-correction code, how many traces do we need to recover it? 
\cite{brakensiek2020coded} shows that the coded trace reconstruction is nearly equivalent to the average-case trace reconstruction problem, and $\e(O(\log^{1/3}(1/\epsilon))$ traces are enough to reconstruct the string. 

The work of \cite{narayanan2021circular} introduces the circular trace reconstruction problem. In this problem, the traces are generated by passing through a deletion channel and then undergoing a random cyclic shift. 
This problem aims to recover the string or any of its cyclic shifts. 
They show, via a reduction, that this problem is at least as hard as the linear trace reconstruction problem and provide a $\e(O(n^{1/3}))$ upper bound when $n$ is a prime or a product of two primes in the worst-case, and a $n^{O_p(1)}$ upper bound for the average case, along with a lower bound of $\tilde{\Omega}(n^3)$. This lower bound comes from showing the indistinguishability of two constant-sparse strings whose cyclic statistics cancel up to order $m=2$; the current paper extends this construction to $m=4$ and shows this is almost the best possible.

\section{Definitions}\label{sec:defs}

We study traces of $k$-sparse binary strings, meaning the set of binary strings with $k$ nonzeros. Throughout this paper, we treat $k$ as a universally bounded constant. For two sequences $x = (x_1,\ldots,x_k)$ and $y = (y_1,\ldots,y_k)$, we say $y$ is a \textit{cyclic shift} of $x$ if there exists an integer $c$ such that $x_j = y_{j+c}$ for all $j \in [k]$, where indices are interpreted mod $k$. If no such $c$ exists, then we say $x$ and $y$ are \textit{cyclically distinct}. 

Every $k$-sparse binary string $x$ is a cyclic shift of some string of the form $10^{x_1}\ldots10^{x_k}$. For our analysis, it will often be useful to represent each string as an integer sequence $x = (x_1,\ldots,x_k)$, where $x_j$ corresponds to the number of zeros in the ``gap'' between the $j^\textrm{th}$ and $j+1^\textrm{st}$ $1$s in $x$. In this paper, we will generally use the integer sequence representation of $x$; when we require the binary string instead, we will explicitly write $x = 10^{x_1}\ldots10^{x_k}$.

We use $[k]$ to refer to the set of integers $\{1,\ldots,k\}$ and $\Bin{z}{q}$ to denote a binomial random variable with $z$ trials and success probability $q$. For two probability measures $\mu$ and $\nu$ over a set $\Omega$, the Hellinger distance between $\mu$ and $\nu$ is defined as $d_H(\mu, \nu) \coloneqq \sqrt{\sum_{x \in \Omega} \left( \sqrt{\mu(x)} - \sqrt{\nu(x)} \right)^2}$.

Most of the technical analysis of this paper concerns understanding properties of ``cyclic statistics'' of sequences $x_1,\ldots,x_k$, which we define here:
\begin{definition}
    Let $i_1,\ldots,i_m \in [k]$, and $\ell$ be a divisor of $k$ (with $\ell=1$ allowed). This defines a \emph{cyclic statistic mod $\ell$} of $x = (x_1,\ldots,x_k) \in \mathbb{Z}^k$ defined by $S_{i_1,\ldots,i_m;\ell}(x) \coloneqq \sum_{j=1}^{k/\ell} x_{i_1 + j\ell}\ldots x_{i_m + j\ell}$, (where all indices are interpreted$\mod k$). 
    When $\ell = 1$, we refer to the function as merely a \textit{cyclic statistic} and write $S_{i_1,\ldots,i_m}(x)$.
\end{definition}

\ifarxiv    
    \paragraph{Problem definitions.}
\else
    \subparagraph*{Problem definitions.} 
\fi
We investigate the problem of recovering a string from its traces generated by a circular deletion channel, which is defined as follows. 

\begin{definition}[Circular deletion channel]
    Fix $p \in (0,1)$. A \textit{circular deletion channel}, denoted $\Del{x}$ takes as input $x \in \{0,1\}^n$ and induces the probability distribution defined by the following process:
    \begin{enumerate}
        \item Delete each bit in $x$ independently with probability $p$ to yield a string $\tilde{x}$. 
        \item Return a uniformly random cyclic shift of $\tilde{x}$.
    \end{enumerate}
\end{definition}
In this paper, we focus on the case where $p$ is a universal constant. We are now ready to define the (circular) trace reconstruction problem. 

\begin{definition}[Trace reconstruction problem]
    Let $x \in \{0,1\}^n$. Given sample access to traces generated by the deletion channel $\Del{x}$, an algorithm $\mathcal{A}$ which returns $x'$ solves the \textit{trace reconstruction problem} if, with probability at least $2/3$, $x'$ is a cyclic shift of $x$. 
\end{definition}

\begin{definition}[Distinguishing problem]
    Let $x,y \in \{0,1\}^n$ with $x$ and $y$ cyclically distinct. An algorithm $\mathcal{A}$ \textit{distinguishes $x$ from $y$} if, given sample access to a deletion channel $\Del{z}$,  $\mathcal{A}$ returns $x$ with probability at least $2/3$ if $x$ is a cyclic shift of $z$, and $\mathcal{A}$ returns $y$ with probability at least $2/3$ if $y$ is a cyclic shift of $z$. 
\end{definition}

\ifarxiv
\begin{remark}
    We note that the two problems above are nearly equivalent for $k$-sparse strings in the sense that an $O(f(n))$ upper bound for the distinguishing problem implies an $O(f(n)\log n)$ algorithm for reconstruction. This is due to the following standard technique. Suppose algorithm $\mathcal{A}$ distinguishes all pairs $x,y \in \mathcal{X}$ with probability at least $2/3$ using $f(n)$ samples. To reconstruct an unknown string $z$, we draw $O(f(n)\log n^{2k}) = O_k(f(n)\log n)$ samples from the deletion channel $\Del{z}$, and test each possible pair of cyclically distinct strings $x,y$ against each other using $\mathcal{A}$. By a Chernoff bound, $\mathcal{A}$ will return a valid answer on each pair with probability $1 - O(1/n^{2k})$. Since there are $O(n^k)$ cyclically distinct $k$-sparse strings of length $n$, \textit{all} $O(n^{2k})$ tests will succeed with constant probability, in which case we return the unique string $z'$ which is chosen by $\mathcal{A}$ against all other candidates. 
\end{remark}
\fi

\section{Cyclic Statistics Characterize Cyclically Distinct Sequences}\label{sec:cyclic}

We begin by studying cyclic statistics of integer sequences, a key component of our upper and lower bounds. The goal of this section is to answer the following question: given two cyclically distinct integer sequences $x$ and $y$, what is the smallest value of $m$ such that there exists a cyclic statistic $S_{i_1,\ldots,i_m}$ satisfying $S_{i_1,\ldots,i_m}(x) \neq S_{i_1,\ldots,i_m}(y)$? We show that for all pairs of sequences, $m \leq 6$ suffices.

Our proof relies on the Fourier transforms of $x$ and $y$, denoted $\hat{x}$ and $\hat{y}$ respectively. We partition the Fourier coefficients $\hat{x}_j$ into equivalence classes based on $\gcd(j,k)$. We show that for each equivalence class, the Fourier coefficients in $x$ are jointly zero or nonzero \emph{together}. Notably, as it turns out, for any $x$ and $y$ with matching second-order cyclic statistics, $\hat{y}$ has identical pattern of zeros and nonzeros as $\hat{x}$ (see Lemma~\ref{lem:zero-or-nonzero}).

Armed with this insight, we restrict our analysis to the set of \emph{nonzero} Fourier coefficients in $x$ and $y$, which correspond to the same set of indices. We define a set of variables $\{z_j = \frac{1}{2\pi i} \log \frac{\hat{x}_j}{\hat{y}_j} \}$ for each index $j$ with nonzero Fourier coefficients in $\hat{x}$ and $\hat{y}$. We show that equivalence of cyclic statistics (up to order 6) imposes a system of linear constraints on the variables $z_j$, and that any set of solutions must be strictly real and  there exists an integer $c$ such that $z_j \equiv \frac{cj}{k} \mod 1$ for all $j$. It follows immediately that $\hat{x}_j = \hat{y}_j \exp\left(2\pi i \cdot \frac{cj}{k}\right)$ for all $j$. By basic properties of Fourier transforms, this implies that $x$ is a cyclic shift of $y$, proving our upper bound.

The organization of this section is as follows. In \Cref{sec:fourier}, we relate cyclic statistics to the Fourier transform and prove that for any $x$ and $y$ with identical cyclic statistics up to order $6$, $\hat{x}$ and $\hat{y}$ have the same pattern of zeros and nonzeros. In \Cref{sec:numtheory} we set up several number theoretic results necessary to analyze the system of linear constraints. Finally, in \Cref{sec:cyclic-main}, we apply these facts to the set of variables $\{z_j\}$ defined above to obtain the main result of this section, Lemma~\ref{lem:cyclic-stats}, which states that any two cyclically distinct strings $x,y$ must differ in \emph{some} cyclic statistic of order $m\leq 6$. 

\subsection{Cyclic Statistics and the Fourier Transform}
\label{sec:fourier}

Consider two integer sequences $x$ and $y$ with identical cyclic statistics up to order $m$. We show that this relationship implies an equivalence between certain $m$-way products of Fourier coefficients.

\begin{lemma}\label{lem:fourier-similar}
    Let $x$ and $y$ be two integer sequences of length $k$ with identical $m\mathrm{th}$ order cyclic statistics. Consider their Fourier transforms $\hat{x}$ and $\hat{y}$. For any $m$-tuple $i_1,\ldots,i_m$ with $i_1+\cdots+i_m\equiv 0 \mod k$, we have $\hat{x}_{i_1}\cdots\hat{x}_{i_m} = \hat{y}_{i_1}\cdots\hat{y}_{i_m}$.
\end{lemma}
\begin{proof}
    Letting $j_2'\equiv j_2-j_1 \mod k \ldots,j_m'\equiv j_m-j_1 \mod k$, and interpreting indices mod $k$:
    \begin{align*}
        \hat{x}_{i_1}\cdots\hat{x}_{i_m}&=\left(\sum_{j_1=0}^{k-1}x_{j_1}e^{\frac{2\pi i}{k}j_1i_1}\right)\cdots\left(\sum_{j_m=0}^{k-1}x_{j_m}e^{\frac{2\pi i}{k}j_mi_m}\right)\\
        &=\left(\sum_{j_1=0}^{k-1}x_{j_1}e^{\frac{2\pi i}{k}j_1i_1}\right)\left(\sum_{j_2'=0}^{k-1}x_{j_2'+j_1}e^{\frac{2\pi i}{k}(j_2'+j_1) i_2}\right)\cdots\left(\sum_{j_m'=0}^{k-1}x_{j_m'+j_1}e^{\frac{2\pi i}{k}(j_m'+j_1) i_m}\right)\\
        &=\sum_{j_2',\ldots,j_m'=0}^{k-1}\left(e^{\frac{2\pi i}{k}(j_2'i_2+\cdots+j_m'i_m)}\sum_{j_1=0}^{k-1}x_{j_1}x_{j_2'+j_1}\ldots x_{j_m'+j_1}\right)\\
        &=\sum_{j_2',\ldots,j_m'=0}^{k-1} e^{\frac{2\pi i}{k}(j_2'i_2+\cdots+j_m'i_m)}S_{0,j_2',\ldots,j_m'}(x) 
    \end{align*}
    where the exponential terms vanish in the second-to-last equality because $i_1+\cdots+i_m\equiv 0 \mod k$ and so $e^{\frac{2\pi i}{k}j_1 i_1}\cdots e^{\frac{2\pi i}{k}j_1 i_m}=1$. Since $S_{0,j_2',\ldots,j_m'}(x) = S_{0,j_2',\ldots,j_m'}(y)$ for all $j_2',\ldots,j_m'$ by assumption, we have that the corresponding $m$-way products of Fourier terms are equal.
\end{proof}

\Cref{lem:fourier-similar} implies that for any $i_1,\ldots,i_m$ satisfying $i_1 + \cdots + i_m \equiv 0 \mod k$, if the Fourier coefficients $\hat{x}_{i_1},\ldots,\hat{x}_{i_m}$ are nonzero, then $\hat{y}_{i_1},\ldots,\hat{y}_{i_m}$ must be nonzero as well. We now partition the elements of  $[k]$ into equivalence classes based on their gcd with $k$.

\begin{definition}
    Let $\alpha$ be a divisor of $k$. We define $G_\alpha \coloneqq \{j \in [k] : \gcd(j,k) = \alpha\}$. 
\end{definition}

We show that for a single integer sequence $x$, the Fourier coefficients of $x$ in each $G_\alpha$ are either \textit{all} zero or \textit{all} nonzero, allowing us to analyze them together. This can be shown using a well-known property of cyclotomic fields: for every primitive $d$th root of unity of form $\omega_d = e^{2\pi i/d}$, the Galois group $\textrm{Gal}(\mathbb{Q}(\omega_d)/\mathbb{Q})$ is isomorphic to $(\mathbb{Z}/d\mathbb{Z})^\times \coloneqq \{j \in \{0,\ldots,d-1\} : \gcd(j,k) = 1\}$, the multiplicative group of integers mod $d$. The isomorphism is specified by \Cref{fact:cyclotomic}. 

\begin{fact}
    \label{fact:cyclotomic}
    Let $d \in \mathbb{N}$ and let $\omega_d \coloneqq e^{2\pi i/d}$ be a primitive $d$th root of unity. Consider $\mathbb{Q}(\omega_d)$, the field extension of $\mathbb{Q}$ generated by $\omega_d$. For all $a \in (\mathbb{Z}/d\mathbb{Z})^\times$, there exists an automorphism $\sigma_a \in \textrm{Gal}(\mathbb{Q}(\omega_d)/\mathbb{Q})$ such that $\sigma_a(\omega_d) = \omega_d^a$.\footnote{For more on cyclotomic fields, see \url{https://en.wikipedia.org/wiki/Cyclotomic_field}.}
\end{fact}

\begin{lemma}
    \label{lem:cyclotomic}
    Let $x$ be an integer sequence of length $k$ with Fourier transform $\hat{x}$. For any $\alpha \in [k]$ and $j,j' \in G_\alpha$, if $\hat{x}_j = 0$, then $\hat{x}_{j'} = 0$. 
\end{lemma}
\begin{proof}
    Suppose $\hat{x}_j = 0$. Let $d \coloneqq \frac{k}{\alpha}$ and consider the a primitive $d$th root of unity $\omega_{d} \coloneqq e^{2\pi i/d}$. The $j$th Fourier coefficient of $x$ is given by 
    \begin{equation*}
        \hat{x}_j = \sum_{\ell=0}^{k-1} x_{\ell} e^{(2\pi i/k)j\ell} 
        = \sum_{\ell=0}^{k-1} x_{\ell} e^{(2\pi i/d) \cdot (d/k) \cdot j\ell}
        = \sum_{\ell=0}^{k-1} x_{\ell} \omega_d^{ (j/\alpha)\ell}
    \end{equation*}
    Let $z$ be a sequence of length $d$ generated by interpreting each index of $x$ mod $d$ and summing together all elements that map to a given location. (In other words, $z_j \coloneqq \sum_{\ell : \ell \equiv j \mod d} x_{\ell}$.) Then the Fourier transform of $z$ at location $\frac{j}{\alpha}$ is given by
    \begin{equation*}
        \hat{z}_{j/\alpha} 
        = \sum_{\ell=0}^{d-1} z_{\ell} \omega_{d}^{(j/\alpha)\ell} 
        = \sum_{\ell=0}^{k-1} x_{\ell} \omega_{d}^{(j/\alpha)(\ell\mod d)} 
        = \sum_{\ell=0}^{k-1} x_{\ell} \omega_{d}^{(j/\alpha)\ell}
        = \hat{x}_j
    \end{equation*}
    Similarly, $\hat{z}_{j'/\alpha} = \hat{x}_{j'}$. Because $z$ is an integer sequence, $\hat{z}_{j/\alpha}$ and $\hat{z}_{j'/\alpha}$ are elements of the cyclotomic field $\mathbb{Q}(\omega_{d/\alpha})$. Choose $b \in (\mathbb{Z}/d\mathbb{Z})^\times$ such that $b \frac{j}{\alpha}\equiv\frac{j'}{\alpha}\mod d$. 
    By \Cref{fact:cyclotomic}, there exists $\sigma_b \in \textrm{Gal}(\mathbb{Q}(\omega_d)/\mathbb{Q})$ such that $\sigma_b(\omega_{d}) = \omega_{d}^b$. Therefore, we have
    \begin{equation*}
        \hat{x}_{j'}
        = \hat{z}_{b \cdot (j/\alpha) \mod d} 
        = \sum_{\ell=0}^{d-1} z_{\ell} \cdot \omega_{d}^{b(j/\alpha)\ell} 
        = \sum_{\ell=0}^{d-1} z_{\ell} \cdot \sigma_b \left( \omega_{d}^{(j/\alpha)\ell} \right) 
        = \sigma_b(\hat{z}_{j/\alpha}) 
        = \sigma_b(\hat{x}_j)
        = 0
    \end{equation*}
    as desired.
\end{proof}

Consider two sequences $x$ and $y$ with identical second order cyclic statistics. As a simple consequence of Lemmas \ref{lem:fourier-similar} and \ref{lem:cyclotomic}, $\hat{x}$ and $\hat{y}$ are jointly zero or jointly nonzero everywhere in the indices belonging to a single $G_\alpha$. 

\begin{lemma}
    \label{lem:zero-or-nonzero}
    Let $x$ and $y$ be two integer sequences of length $k$ with identical second order cyclic statistics.
    Then for each $\alpha \in [k]$, either (1) $\hat{x}_j \neq 0$ and $\hat{y}_j \neq 0$ for all $j \in G_\alpha$ or (2) $\hat{x}_j = \hat{y}_j = 0$ for all $j \in G_\alpha$. 
\end{lemma}
\begin{proof}
    Suppose $\hat{x}_j \neq 0$ for some $j \in G_\alpha$. By \Cref{lem:cyclotomic}, we must have $\hat{x}_{j'} \neq 0$ for all $j' \in G_\alpha$. Note that $k-j \in S_{\alpha}$, so $\hat{x}_{k-j} \neq 0$. Moreover, $j + (k-j) \equiv 0 \mod k$, so by \Cref{lem:fourier-similar}, we have $\hat{y}_j \cdot \hat{y}_{k-j} = \hat{x}_j \cdot \hat{x}_{k-j} \neq 0$. This requires that $\hat{y}_j \neq 0$. Applying \Cref{lem:cyclotomic} to $y$, we have that $\hat{y}$ is nonzero for all indices in $G_\alpha$, proving the first statement. 

    If $x_j = 0$, then similarly $\hat{y}_j \cdot \hat{y}_{k-j} = \hat{x}_j \cdot \hat{x}_{k-j} = 0$, requiring that either $\hat{y}_j = 0$ or $\hat{y}_{k-j} = 0$. By \Cref{lem:cyclotomic}, $\hat{x}$ and $\hat{y}$ must be zero everywhere in $G_\alpha$, proving the second statement.
\end{proof}

\subsection{Number Theoretic Prerequisites}
\label{sec:numtheory}

Recall that our goal is to analyze a system of constraints involving the nonzero Fourier coefficients of $x$ and $y$, where the constraints arise from the assumption that $x$ and $y$ have identical cyclic statistics up to order $6$. These constraints will turn out to take the form $z_{j_1} + \cdots + z_{j_m} \in \mathbb{Z}$, for a certain set of complex valued variables $\{z_j\}$ and certain $m$ tuples $j_1,\ldots,j_m$. We wish to show that any set of variables satisfying these constraints must be strictly real and form an arithmetic sequence mod $1$. In \Cref{lem:mod-relations}, we show this is true; we apply this fact in \Cref{sec:cyclic-main} to show that $x$ and $y$ are cyclically equivalent. 

Before proving this relation, we require certain number theoretic facts. First, we show that for any integer $d$ and $j \in \mathbb{Z}/d\mathbb{Z}$, $j$ can be written as a sum (mod $d$) of either 2 or 3 numbers relatively prime to $d$. Here and below, we use $\mathbb{Z}/d\mathbb{Z}$ to refer to the cyclic group of integers modulo $d$, represented as the set $\{0,\ldots,d-1\}$. 

\begin{fact}\label{fact:2-or-3}
    The following facts are true for any prime power $p^a$ and integer $j$:
    \begin{itemize}
        \item If $p \neq 2$, there exists a pair $b_1,b_2 \in [p^a]$ that are relatively prime to $p$ such that $b_1+b_2\equiv j \mod p^a$. Also, there exists a triple $b_1,b_2,b_3 \in [p^a]$ that are relatively prime to $p$ such that $b_1+b_2+b_3\equiv j \mod p^a$.
    
        \item If $p = 2$ and $j$ is even, there exists a pair $b_1,b_2 \in [p^a]$ that are relatively prime to $p$ such that $b_1+b_2\equiv j \mod p^a$. 
    
        \item If $p = 2$ and $j$ is odd then there exists a triple $b_1,b_2,b_3 \in [p^a]$ that are relatively prime to $p$ such that $b_1+b_2+b_3\equiv j \mod p^a$.
    \end{itemize}
\end{fact}
\begin{proof}

    First, consider the case in which $p \neq 2$. At least one of $j-1,j+1$ must be relatively prime to $p^a$ since their difference is 2, which itself is relatively prime to $p$. Since $j\equiv (j-1)+1 \mod p^a $, $j\equiv (j+1)+(p^a-1) \mod p^a$, and both $1$ and $p^a-1$ are relatively prime to $p^a$, at least one of these two expressions expresses $j$ as the sum of 2 numbers both relatively prime to $p^a$.
    Similarly, at least one of $j-2,j+2$ must be relatively prime to $p^a$ since their difference is 4, which itself is relatively prime to $p$. 
    Thus since $j\equiv (j-2)+1+1 \mod p^a$ and $j\equiv (j+2)+(p^a-1)+(p^a-1) \mod p^a$, and since both $1$ and $p^a-1$ are relatively prime to $p^a$, at least one of these two expressions expresses $j$ as the sum of 3 numbers both relatively prime to $p^a$.
    
    For the case of $p=2$, we note that if $j$ is even, then $j-1$ is relatively prime to $2^a$ and thus we express $j= (j-1)+1$. Otherwise, if $j$ is odd, then $j-2$ is relatively prime to $2^a$, and we express $j=(j-2)+1+1$.
\end{proof}

\begin{lemma}
    \label{lem:mod-representation}
    Let $d \in \mathbb{N}$. The following is true for all $j \in \mathbb{Z}/d\mathbb{Z}$:
    \begin{itemize}
        \item If $d$ is even and $j$ is odd, then $j$ can be written as the sum mod $d$ of 3 numbers relatively prime to $d$. 
    
        \item Else 
        $j$ can be written as the sum mod $d$ of 2 numbers relatively prime to $d$. 
    \end{itemize}
\end{lemma}

\begin{proof}
    Consider the prime factorization of $d$: $d=p_1^{a_1}\cdot\ldots\cdot p_\ell^{a_\ell}$. Given $j\in \mathbb{Z}/d\mathbb{Z}$, we use the Chinese remainder theorem to express $j$ via its residues mod $p_i^{a_i}$, for each prime factor $p_i$. Define $j_i \coloneqq j\mod p_i^{a_i}$. We note that $j$ is relatively prime to $d$ if and only if each $j_i$ is relatively prime to its corresponding prime power $p_i^{a_i}$.

    We use Fact~\ref{fact:2-or-3} (and the Chinese remainder theorem) to canonically represent each $j\in \mathbb{Z}/d\mathbb{Z}$ as the sum mod $d$ of either 2 or 3 numbers that are each relatively prime to $d$. If either $d$ is odd or $j$ and $d$ are both even, we represent $j$ as the sum of 2 numbers, by applying Fact~\ref{fact:2-or-3} to represent each $j_i$ as $j_i \equiv b_i^{(1)} + b_i^{(2)} \mod p_i^{a_i}$ with $b_i^{(1)}$ and $b_i^{(2)}$ relatively prime to $p_i^{a_i}$. By the Chinese remainder theorem, the system, $b^{(1)} \equiv b_i^{(1)} \mod p_i^{a_i}, \forall i \in [\ell]$, has a unique solution $b^{(1)}$ relatively prime to $d$; similarly, there exists a unique $b^{(2)}$ relatively prime to $d$ satisfying $b^{(2)} \equiv b_i^{(2)} \mod p_i^{a_i}$ for all $i$. Finally, $b^{(1)} + b^{(2)} \equiv j \mod d$, which proves the first case. If $j$ is odd and $d$ is even, we use the analogous procedure to represent $j$ as the sum of 3 numbers relatively prime to $d$, which proves the second case. 
\end{proof}

Using the representation given by \Cref{lem:mod-representation}, we analyze the system of complex-valued linear equations described at the beginning of this subsection.

\begin{lemma}\label{lem:mod-relations}
Let $d \in \mathbb{N}$ and $z_j \in \mathbb{C}$ for each $j \in [d]$ relatively prime to $d$. Suppose that for every $m \le 6$, and any $m$-tuple $(j_1, \ldots, j_m)$ relatively prime to $d$ satisfying $j_1 + \cdots + j_m \equiv 0 \pmod d$, we have $z_{j_1} + \cdots + z_{j_m} \in \mathbb{Z}$. Then there exists an integer $c$ such that for all $j$ relatively prime to $d$ we have $z_j - \frac{cj}{d} \in \mathbb{Z}$.
\end{lemma}
\begin{proof}
For all $j \in \mathbb{Z}/d\mathbb{Z}$, let $[[j]]$ denote the canonical representation of $j$ given by \Cref{lem:mod-representation} as the sum mod $d$ of either 2 or 3 numbers relatively prime to $d$ (so $[[j]]$ is a set of size 2 or 3). Let $z_{[[j]]}$ denote the sum of the 2 or 3 corresponding variables. For convenience, we write the set $d - [[j]] \coloneqq \{d - j': j' \in [[j]]\}$. Then the conditions of the lemma imply the following set of conditions:
\begin{align*}
    &\forall j \text{ relatively prime to } d:\;z_j+z_{d-j}\in \mathbb{Z}, \\
    &\forall j \text{ relatively prime to } d:\;z_{[[j]]}+z_{d-j}\in \mathbb{Z}, \\
    &\forall j\in \mathbb{Z}/d\mathbb{Z}:\; z_{[[j]]}+z_1+z_{d-[[j+1]]} \in \mathbb{Z}
\end{align*} 
To see the implication, it suffices to show that the left hand side of each constraint is a sum of at most 6 real variables whose indices are relatively prime to $d$. For all $j$ with $\gcd(j,d) = 1$, we have $\gcd(d-j,d) = 1$, so the first constraint satisfies the condition. For all $j$, $z_{[[j]]}$ can be written as a sum of 2 or 3 variables whose indices are relatively prime to $d$, so the second constraint requires at most 4 variables. Finally, if $d$ is odd, the representation $[[j]]$ comprises 2 variables for all $j$, in which case the third constraint requires 5 variables. If $d$ is even, the values $j$ and $j+1$ comprise 1 even number and 1 odd number, and $d - [[j+1]]$ has the same parity as $j+1$, resulting in at most 6 terms. 

Note that these three sets of conditions collectively imply that $z_{[[j]]}+z_1- z_{[[j+1]]} \in\mathbb{Z}$ for all $j \in \mathbb{Z}/d\mathbb{Z}$. Therefore, any set of solutions $z_{[[j]]}$ must form an arithmetic sequence of increment $z_1$ plus some integer, i.e., $z_{[[j]]} - jz_1 \in\mathbb{Z}$ for all $j$. Therefore, $dz_1 \in \mathbb{Z}$, which requires that $z_1$ is a multiple of $\frac{1}{d}$ (and is therefore real). Writing $z_1 = \frac{c}{d}$ for some integer $c$, we have $z_{[[j]]} \equiv \frac{cj}{d} \mod 1$ for all $j \in \mathbb{Z}/d\mathbb{Z}$. Finally, the first and second conditions imply that $z_j = z_{[[j]]}$ for all $j$ relatively prime to $d$, which proves the lemma. 
\end{proof}

The following technical lemma is used in the proof of the main result of this section, Lemma~\ref{lem:cyclic-stats}. Intuitively, in the proof of Lemma~\ref{lem:cyclic-stats} we aim to show that variables $z_0,\ldots,z_{k-1}$ form an arithmetic sequence. Instead we show that, for each value $\alpha$, those $z_j$ with index $j$ satisfying $gcd(j,k)=\alpha$ form an arithmetic sequence of increment $c_\alpha$. We use the following lemma to show that this implies there is a consistent arithmetic sequence across all $j$, with some increment $c$.

\begin{definition}
    For a positive integer $a \in \mathbb{N}$ and prime number $p$, define the \textit{$p$-adic valuation} of $a$ to be  $\nu_p(a) \coloneqq \max\{\ell \in \mathbb{Z}_{\geq 0} : p^{\ell} \mid a\}$. 
\end{definition}

\begin{lemma}
\label{lem:unique-c}
    Let $k \in \mathbb{N}$ and $A \subseteq [k]$. Let $c_\alpha \in \mathbb{Z}$ for each $\alpha \in A$, and suppose the following are true for all pairs $\alpha, \alpha' \in A$:
    \begin{enumerate}
        \item $\lcm(\alpha,\alpha') \in A$.
        \item $(c_{\alpha'} - c_{\alpha}) \cdot \lcm(\alpha,\alpha') \equiv 0 \mod k$
    \end{enumerate}
    Then there exists $c \in \mathbb{Z}$ such that $(c - c_\alpha)\alpha \equiv 0 \mod k$ for all $\alpha \in A$. 
\end{lemma}
\begin{proof}
    Consider the prime factorization of $k$, $k = p_1^{a_1}\ldots p_{\ell}^{a_{\ell}}$. For each $i \in [\ell]$, we will show there exists $c_i$ satisfying the desired property mod $p_i^{a_i}$, and conclude the lemma by the Chinese remainder theorem. 
    
    Fix $i \in [\ell]$. Define $\alpha_i$ to be the element of $A$ with the lowest $p_i$-adic valuation: $\alpha_i = \argmin_{\alpha \in A} \nu_{p_i}(\alpha)$. Let $c_i$ be an integer satisfying $(c_i - c_{\alpha_i})\alpha_i \equiv 0 \mod p_i^{a_i}$, and consider any $\alpha' \in A$. By a well-known property of $p$-adic valuations, we have $\nu_{p_i}(\lcm(\alpha',\alpha_i)) = \max\{\nu_{p_i}(\alpha'), \nu_{p_i}(\alpha_i)\} = \nu_{p_i}(\alpha')$. The prime power $p_i^{\nu_{p_i}(\alpha')}$ divides $\alpha'$, so it also divides $\lcm(\alpha', \alpha_i)$. Therefore, we have the following two facts:
    \begin{align*}
        (c_{\alpha_i} - c_{\alpha'})p^{{\nu_{p_i}(\alpha')}} &\equiv 0 \mod p_i^{a_i} \tag{by assumption (2) of the lemma} \\
        (c_i - c_{\alpha_i}) p_i^{\nu_{p_i}(\alpha_i)} &\equiv 0 \mod p_i^{a_i} \tag{by definition of $c_i$}
    \end{align*}
    Using the fact that $\nu_{p_i}(\alpha') \geq \nu_{p_i}(\alpha_i)$ to combine the two equations, we have $(c_i - c_{\alpha'}) p_i^{\nu_{p_i}(\alpha')} \equiv 0 \mod p_i^{a_i}$. 
    Finally, we multiply both sides by the integer $\alpha'p_i^{-\nu_{p_i}(\alpha')}$ to yield $(c_i - c_{\alpha'}) \alpha' \equiv 0 \mod p_i^{a_i}$. 
    Choosing $c_i = c_{\alpha_i}$ for all $i$, the equation above holds for all pairs $i,\alpha$. We now wish to find $c$ such that $c \equiv c_i \mod p_i^{a_i}$ for all $i$. Such an integer $c$ is guaranteed to exist by the Chinese remainder theorem, which completes the proof.
\end{proof}

\subsection{Main Characterization}
\label{sec:cyclic-main}

We now prove the main result of this section. 

\begin{lemma}
    \label{lem:cyclic-stats}
    Let $x$ and $y$ be two integer sequences of length $k$ with identical cyclic statistics up to order $6$. Then $x$ and $y$ must be identical up to a cyclic shift. 
\end{lemma}
\begin{proof}
    It suffices to show that there exists an integer $c$ such that $\hat{x}_j = \hat{y}_j \cdot \exp\left(2\pi i \cdot \frac{cj}{k}\right)$ for all $j \in [k]$. By \Cref{lem:zero-or-nonzero}, this property holds for all $j$ where $\hat{x}_j = 0$. Therefore, we consider only the divisors $\alpha$ such that the Fourier coefficients are nonzero everywhere in $G_\alpha$ for both $x$ and $y$. For all $j$ with nonzero Fourier coefficients $\hat{x}_j$ and $\hat{y}_j$, define $z_j \coloneqq \frac{1}{2\pi i} \cdot \log \left( \frac{\hat{x}_j}{\hat{y}_j} \right)$, where $\log$ is the unique complex-valued function satisfying $e^{\log z} = z$ and $\textrm{Im} \log z \in (0,2\pi]$ for all $z \in \mathbb{C} \backslash \{0\}$. To prove the lemma, we show that there exists an integer $c$ such that for all $j$ with $\hat{x}_j \neq 0$ and $\hat{y}_j \neq 0$, $\textrm{Im} z_j = 0$ and $z_j \equiv \frac{cj}{k} \mod 1$. 
    
    Fix $\alpha \in [k]$ and let $d \coloneqq \frac{k}{\alpha}$. Note that for all $j \in S_{\alpha}$, $\gcd\left(\frac{j}{\alpha}, d\right) = 1$. Therefore, each $j \in G_\alpha$ corresponds to a unique number relatively prime to $d$. Consider any $m$-tuple $j_1,\ldots,j_m \in G_\alpha$ such that $\frac{j_1}{\alpha} + \ldots + \frac{j_m}{\alpha} \equiv 0 \mod d$. By \Cref{lem:fourier-similar}, we have $z_{j_1} + \ldots + z_{j_m} \in \mathbb{Z}$. Applying \Cref{lem:mod-relations} (and using the fact that $k = d\alpha$), there exists an integer $c_\alpha$ such that for all $j \in G_\alpha$, we have $z_j - \frac{c_{\alpha}j}{k} \in \mathbb{Z}$. In other words, $z_j$ is real and $z_j \equiv \frac{c_{\alpha}j}{k} \mod 1$. 

    We have shown that for all $G_\alpha$ with nonzero Fourier coefficients, we can write $kz_j \equiv c_\alpha j \mod k$. We will now show that this is true for a \textit{single} consistent integer $c = c_\alpha$ across all $\alpha$. Consider $\alpha\neq\alpha'$ where $\hat{x}$ and $\hat{y}$ are nonzero everywhere in $G_\alpha \cup S_{\alpha'}$. Let $c_\alpha,c_{\alpha'}$ respectively be the corresponding coefficients, and $\ell=\lcm(\alpha,\alpha')$. Since $\ell$ is a multiple of $\alpha$, \Cref{lem:mod-representation} implies that we can write $\ell/\alpha$ as a sum mod $k/\alpha$ of 2 or 3 numbers relatively prime to $k/\alpha$. Multiplying by $\alpha$, we can write $\ell$ as a sum mod $k$ of 2 or 3 numbers whose gcd with $k$ is $\alpha$; denote this representation as $[[\ell]]$. Since $\ell$ is also a multiple of $\alpha'$, it must also have a representation $[[\ell]]'$ as a sum mod $k$ of 2 or 3 numbers whose gcd with $k$ is $\alpha'$.
    
    We write $z_{[[\ell]]}$ and $z_{[[\ell']]}$ to denote the sum of the variables corresponding to the indices in $[[\ell]]$ and $[[\ell']]$ respectively, and use the convention that $k-[[\ell]] = \{k - j : j \in [[\ell]]\}$. Since $\frac{j}{\alpha} + \frac{k-j}{\alpha} \equiv 0 \mod \frac{k}{\alpha}$ for all $j \in G_\alpha$, we have $z_{\ell} + z_{k-[[\ell]]} \equiv 0 \mod 1$ for all $\ell$. From the previous paragraph, we have that the sum of the elements in $[[\ell]]$ equals the sum of the elements in $[[\ell]]'$ mod $k$. We thus conclude the condition that $\ell$ has a \emph{consistent} coefficient in both representations, as $z_{[[\ell]]}=z_{[[\ell]]'}$, and reexpress this as a sum of 4 to 6 coefficients, moving everything to the left-hand side, as $z_{[[\ell]]}+z_{k-[[\ell]]'}\equiv 0 \mod 1$. Recall that $[[\ell]]$ only contains indices $j$ with $\gcd(j,k)=\alpha$, and the corresponding variables $z_j$ satisfy $kz_j \equiv c_\alpha j \mod k$. Similarly, $[[\ell]]'$ only contains indices $j'$ with $\gcd(j',k)=\alpha'$, and each $z_{j'}$ satisfies $kz_{j'} \equiv c_{\alpha'} j' \mod k$. Therefore, the constraint $(c_{\alpha'}-c_{\alpha})\ell \equiv 0 \mod k$ holds for any pair of $\alpha,\alpha' \in A \subseteq [k]$, where $A = \{\alpha: \hat{x},\hat{y} \text{ are nonzero everywhere in }S_{\alpha}\}$. 
    
    Applying \Cref{lem:unique-c}, there exists $c\in\mathbb{Z}$, such that $(c-c_{\alpha})\alpha \equiv 0 \mod k$ for all $\alpha \in A$. 
    Therefore, for every $\alpha \in A$, for every $j\in S_{\alpha}$, $c\alpha \frac{j}{\alpha} \equiv c_{\alpha} \alpha \frac{j}{\alpha} \mod k$. 
    Notice that by definition, $\frac{j}{\alpha}$ is an integer and thus, for every $\alpha \in A$, for every $j\in S_{\alpha}$, $cj \equiv c_{\alpha} j \mod k$. 
    Therefore, for every $j$ where $\hat{x}_j$, $\hat{y}_j$ are nonzero, we have $kz_j \equiv c_{\alpha} j \equiv c j \mod k$. 
    Dividing by $k$, we have $z_j \equiv \frac{cj}{k} \mod 1$ for every $j$ where $\hat{x}_j$, $\hat{y}_j$ are nonzero, as desired.
\end{proof}

\section{Upper Bound}
\label{sec:upperbound}

We now give an algorithm for distinguishing any two cyclically distinct strings with a constant number of $1$s. Given \Cref{lem:cyclic-stats}, a naïve approach for testing whether a set of traces is generated from $x$ or $y$ is to identify a cyclic statistic in which $x$ and $y$ differ, and then estimate this statistic from traces. For a string with $k$ $1$s, the deletion channel will preserve all $1$s with constant probability; a resulting trace can be represented as a sequence of $k$ ``gaps,'' for which we can compute cyclic statistics. In particular, for a string $x = (x_1,...,x_k)$ and retention probability $q:=1-p$, a trace can be represented as a sequence of binomial random variables $\tx = (\tx_1,...,\tx_k)$, where $\tx_j \sim \Bin{x_j}{q}$. 

Unfortunately, computing a cyclic statistic of $\tx$ requires multiplying up to 6 of these binomial random variables. In the worst case, we might have $x_j = \Omega(n)$ for all $j$, causing our estimator to have variance roughly $n^{11}$. Since cyclic statistics are integers, we may need to estimate the desired cyclic statistic to within constant error, and thus the naïve algorithm would require $\Omega(n^{11})$ samples.

We give a more efficient algorithm using the following insight: while the product of $6$ independent binomials with $\Omega(n)$ trials each has variance roughly $\Omega(n^{11})$, the variance is considerably lower if each binomial is shifted to have mean $\tilde{O}(\sqrt{n})$. To accomplish this, we preprocess $x$ and $y$ by greedily grouping the values $qx_1,\ldots,qx_k,qy_1,\ldots,qy_k$ that are within $\tilde{O}(\sqrt{n})$ of each other into a cluster. When processing a trace $\tz$ from either $x$ or $y$, we compute the cluster center closest to each element $\tz_j$, and subtract that center from $\tz_j$. Because the gap sizes $\tz_j$ are highly likely to be within $\tilde{O}(\sqrt{n})$ of their expectations, we can find the ``correct'' cluster with high probability. The resulting estimator has variance $\tilde{O}(n^6)$. 

Of course, subtracting an arbitrary sequence $s$ from the observed trace $\tz$ raises another concern: the sequences $qx - s$ and $qy - s$ might now be cyclically equivalent, erasing any differences in their cyclic statistics. We circumvent this problem with a slightly different estimator which only considers cyclic statistics \emph{mod $\ell$} where $\ell$ is chosen so that $s$ is preserved by a cyclic shift of $\ell$.


The remainder of this section is organized as follows. In \Cref{sec:centers}, we give an algorithm for generating a well separated partition of the expectations $qx_1,\ldots,qx_k,qy_1,\ldots,qy_k$. 
In \Cref{sec:cyclic-stats-ub}, we show that cyclic statistics mod $\ell$ are preserved by subtraction of a sequence $s$ invariant to cyclic shifts of $\ell$. Finally, in \Cref{sec:algorithm}, we present our cyclic statistics-based algorithm and prove its correctness. 

\subsection{Determining Centers}
\label{sec:centers}

As a preprocessing step, our algorithm partitions the means of $2k$ binomial random variables (the gaps in traces from $x$ and $y$) into well separated clusters. For our purposes, clusters are considered to be well separated if they are $\Omega(\sqrt{n}\log n)$ far apart; we give a formal definition below.

\begin{definition}
    \label{def:partition}
    Let $x_1,\ldots,x_k \in \mathbb{R}$ and $C >0$. Then a function $c: \mathbb{R} \rightarrow [k]$ is a \textit{$C$-separated partition} of $x_1,\ldots,x_k$ if it satisfies the following:
    \begin{enumerate}
        \item For any $j,j' \in [k]$, if $c(x_j) \neq c(x_{j'})$, then $|qx_{j} - qx_{j'}| > 2C \sqrt{n} \log n$.
        
        \item For any $j,j' \in [k]$, if $c(x_j) = c(x_{j'})$, then $|qx_{j} - qx_{j'}| \leq 2Ck\sqrt{n} \log n$. 

        \item For all $x \in \mathbb{R}$, $c(x) = c\left(\argmin_{x_j:j \in [k]} \abs{x_j - x}\right)$. 
    \end{enumerate}
\end{definition}

We show that a $C$-separated partition exists for any $C > 0$ and $x_1,\ldots,x_k \in \mathbb{R}$. The proof is constructive: we give a greedy algorithm 
\ifarxiv
(\Cref{alg:partition}) 
\fi
for producing a $C$-separated partition of a fixed set of points $x_1,\ldots,x_k$. The algorithm starts by assigning each point to its own cluster; it then iteratively merges clusters which violate the separation condition. When no violations remain, we prove that the algorithm has obtained a $C$-separated partition.  

\ifarxiv
\begin{algorithm}[tbph]
\caption{\textsc{Partition}}\label{alg:partition}
\KwIn{$x_1,\ldots,x_k \in \mathbb{R}; C > 0$}
\KwOut{$C$-separated partition $c$ of $x_1,\ldots,x_k$} 

\DontPrintSemicolon
\nl $S_i \gets \{x_i\}$, $\forall i \in [k]$ \label{line:init} \tcp{initialize clusters with 1 element}
\nl \While{$\exists i, j \in [k] : i \neq j \land \min_{x \in S_i,\; x' \in S_j} |x - x'| \leq 2C\sqrt{n} \log n$} {
    \nl $S_i \gets S_i \cup S_j$ \tcp{merge clusters which are not $C$-separated}
    \nl $S_j \gets \emptyset$ \label{line:emptyset} \;
}
\nl \textbf{foreach }$i \in [k]$ \textbf{do} $c(x) \gets i, \ \forall x \in S_i$ \label{line:part1}\;

\nl Define $c(x) = c(\argmin_{x_j:j \in [k]} \abs{x_j - x}), \forall x \in \mathbb{R} \backslash \{x_1,\ldots,x_k\}$ \label{line:part2} \;
\nl \Return $c$ \;
\end{algorithm}
\fi

\begin{lemma} 
    \label{lem:partition}
    Let $x_1,\ldots,x_k \in [n]$, and $C > 0$. Then there exists a $C$-separated partition of $x_1,\ldots,x_k$. 
\end{lemma}
\ifarxiv
\begin{proof}
    We show that \Cref{alg:partition} returns a valid partition. First, notice that at each iteration, the number of non-empty sets decreases by $1$ (\Cref{line:emptyset}), thus the algorithm terminates. 
    When the while loop terminates, all sets are separated by at least $2C\sqrt{n}\log n$, so the first condition in  \Cref{def:partition} is satisfied. 

    We prove the second condition by induction over the number of iterations $t$: for all $i \in [k]$ and $x,x' \in S_i$, we have $|x-x'| \leq 2C(|S_i|-1)\sqrt{n}\log n$. For $t=0$, the claim is trivial. Suppose the claim holds after the $t$th iteration of the while loop (with $t > 0$). If the loop does not terminate after this iteration, there must exist some $i,j$ and $x \in S_i, x' \in S_j$ with $|x - x'| \leq 2C\sqrt{n}\log n$. Since $|S_i|$ does not decrease, the claim still holds for any $y,y' \in S_i$ (equivalently, $y,y' \in S_j$) by the inductive hypothesis. For $y \in S_i$ and $y' \in S_j$, we have
    \begin{align*}
        |y-y'| &\leq |y-x| + |x-x'| + |x'-y'| \\
        &= 2C(|S_i|-1)\sqrt{n}\log n + 2C\sqrt{n}\log n + 2C(|S_j|-1)\sqrt{n}\log n \\
        &= 2C(|S_i \cup S_j|-1)\sqrt{n}\log n.
    \end{align*}
    
    After the loop terminates, $|S_i| \leq k$ for all $i$. Applying the claim, for any $x,x' \in S_i$, we must have $|x-x'| \leq 2C(|S_i| - 1)\sqrt{n} \leq 2Ck\sqrt{n}\log n$. Therefore, for all $j,j' \in [k]$, $c(x_j) = c(x_{j'})$ if and only if $|x_j - x_{j'}| \leq 2Ck\sqrt{n} \log n$. 

    Finally, the third condition holds by construction of $c$ (\Cref{line:part2}), which completes the proof.
\end{proof}
\fi

\subsection{Uniqueness of Cyclic Statistics}
\label{sec:cyclic-stats-ub}


As discussed at the beginning of this section, our goal is to reduce the variance of our estimator by shifting each binomial random variable $\tx_1,\ldots,\tx_k$ by a value which is within a constant of its mean. In particular, let $c$ be a $C$-separated partition of $qx_1,\ldots,qx_k,qy_1,\ldots,qy_k$, and $s = (c(qx_1),\ldots,c(qx_k))$. The binomial random variables concentrate tightly around their expectation, allowing us to approximately recover the partition of $x$ and $y$ generated by 
\ifarxiv
\Cref{alg:partition} 
\else 
the partitioning algorithm
\fi 
with high probability. The logical approach is therefore to estimate cyclic statistics of $x-s$. 

Unfortunately, this introduces a second problem. If $x$ and $y$ are distinct after applying the partition $c$ elementwise, then they can be distinguished simply by recovering the partition. But if $x$ and $y$ appear similar under the partition --- that is, $s_x = (c(qx_1),\ldots,c(qx_k))$ is a cyclic shift of $s_y =(c(qy_1),\ldots,c(qy_k))$ --- then $x-s_x$ might be a cyclic shift of $y-s_y$, resulting in identical cyclic statistics and thus a failure to distinguish $x$ from $y$.

\begin{example}\label{ex:cyclic}
Consider the case where $k=6$, and where the sequence of shifts/clusters $s_x=s_y$ follows the pattern $s,s',s',s,s',s'$. Further, let $x=(\frac{s}{q}+0,\frac{s'}{q}+1,\frac{s'}{q}+2,\frac{s}{q}+0,\frac{s'}{q}+1,\frac{s'}{q}+2)$ and $y=(\frac{s}{q}+1,\frac{s'}{q}+2,\frac{s'}{q}+0,\frac{s}{q}+1,\frac{s'}{q}+2,\frac{s'}{q}+0)$. In this case, when we sample from $\Del{x}$ or $\Del{y}$, once we subtract $s_x=s_y$, we get two 6-tuples that are cyclically identical; however, algorithmically, we would hope to be able to distinguish traces from $x$ vs $y$, since the \emph{alignment} between the sequence of shifts $s_x=s_y$ vs the sequence of offsets $(0,1,2,0,1,2)$ differs for $x$ vs $y$. The right thing to do here, after subtracting $s_x=s_y$, is to look at \emph{cyclic statistics mod 3}.
\end{example}

This motivates our general strategy of looking at cyclic statistics mod $\ell$, where $\ell$ is the smallest period of repetition of the shift sequence.


\begin{definition}
    A sequence $s = (s_1,\ldots,s_k) \in \mathbb{R}^k$ is \textit{symmetric mod $\ell$} if for all $j \in [k]$, $s_j = s_{(j + \ell) \mod k}$. 
\end{definition}

\begin{definition}
    Let $x = (x_1,\ldots,x_k) \in \mathbb{Z}^k$, $i_1,\ldots,i_m \in [k]$, and $\ell$ be a divisor of $k$. For a sequence $s = (s_1,\ldots,s_k) \in \mathbb{R}^k$, the function $S_{i_1,\ldots,i_m;\ell}(x-s) \coloneqq \sum_{j=1}^{k/\ell} (x_{i_1 + j\ell} - s_{i_1 + j\ell}) \cdots (x_{i_m + j\ell} - s_{i_m + j\ell})$ is an $m$-th order cyclic statistic mod $\ell$, \textit{shifted by $s$}. 
\end{definition}

\ifarxiv
    We caution that for $\ell > 1$, an arbitrary statistic $S_{i_1,\ldots,i_m;\ell}(x-s)$ is \emph{not} necessarily invariant under cyclic shifts of $x$. However, if $s$ is symmetric mod $\ell$, then $S_{i_1,\ldots,i_m;\ell}(x-s)$ \emph{is} invariant under cyclic shifts which are multiples of $\ell$. This turns out to be sufficient: in \Cref{lem:stats-diff}, we prove that if $x$ and $y$ are cyclically distinct integer sequences and $s$ is a sequence which is symmetric mod $\ell$, $x-s$ and $y-s$ must differ in \emph{some} cyclic statistic mod $\ell$. We begin by showing that this difference holds for unshifted sequences. 
\else
    We show that for any period $\ell$ and set of shifts $s$ that repeats mod $\ell$, there is an $\leq 6$-th order statistic that differs on $x,y$, when shifted by $s$ and taken mod $\ell$. We prove this by starting with \Cref{lem:cyclic-stats} that guarantees the existence of a statistic of order $\leq 6$ that distinguishes $x$ from $y$; we show this implies existence of a distinguishing statistic of the same order mod $\ell$. Taking the distinguishing statistic of \emph{smallest} order lets us subtract off any sequence $s$ that is symmetric mod $\ell$, without modifying the discrepancy between $x$ and $y$. This yields:
\fi

\ifarxiv
\begin{lemma}
    \label{lem:cyclic-stats-mod}
    Let $x$ and $y$ be two cyclically distinct integer sequences of length $k$. Then for any divisor $\ell$ of $k$, there exists $m \leq 6$ and a sequence $i_1,\ldots,i_m \in [k]$ such that $S_{i_1,\ldots,i_m;\ell}(x) \neq S_{i_1,\ldots,i_m;\ell}(y)$. 
\end{lemma}
\begin{proof}
    We prove the contrapositive. Suppose that $S_{i_1,\ldots,i_m;\ell}(x) = S_{i_1,\ldots,i_m;\ell}(y)$ for all $i_1,\ldots,i_m$ for all $m\leq 6$. For any $i_1,\ldots,i_m \in [k]$, the ``mod 1'' statistic $S_{i_1,\ldots,i_m;1}$ can be written as a sum of $m$th order cyclic statistics mod $\ell$:
    \begin{align*}
        S_{i_1,\ldots,i_m;1}(x) &= \sum_{j=1}^k x_{i_1+j}\ldots x_{i_m+j} = \sum_{j=1}^{\ell} \sum_{j'=1}^{k/\ell}x_{i_1+j+\ell j'}\ldots x_{i_m+j+ \ell j'}\\
        &= \sum_{j=1}^{\ell} S_{i_1+j,\ldots,i_m+j;\ell}(x) = \sum_{j=1}^{\ell} S_{i_1+j,\ldots,i_m+j;\ell}(y) 
        = S_{i_1,\ldots,i_m;1}(y)
    \end{align*}
    where all indices are interpreted mod $k$ (and the second to last equality is by assumption). Since the above equality holds for all $m \leq 6$, \Cref{lem:cyclic-stats} implies that $x$ and $y$ are identical up to a cyclic shift. Taking the contrapositive yields the lemma.
\end{proof}
\fi


\begin{lemma}
\label{lem:stats-diff}
    Let $x$ and $y$ be two cyclically distinct integer sequences of length $k$. Let $\ell$ divide $k$, and let $s = (s_1,\ldots,s_k) \in \mathbb{R}^k$ be a sequence that is symmetric mod $\ell$. Then there exists $i_1,\ldots,i_m \in [k]$ with $m\leq 6$ such that \[\abs{S_{i_1,\ldots,i_m;\ell}(x-s) - S_{i_1,\ldots,i_m;\ell}(y-s)}=\abs{S_{i_1,\ldots,i_m;\ell}(x) - S_{i_1,\ldots,i_m;\ell}(y)} \geq 1.\]
\end{lemma}
\ifarxiv
\begin{proof}
    Let $S_{i_1,\ldots,i_m}$ be an order $m$ cyclic statistic such that (1) $S_{i_1,\ldots,i_m;\ell}(x) \neq S_{i_1,\ldots,i_m;\ell}(y)$ and (2) $x$ and $y$ have identical order $m'$ cyclic statistics for all $m' < m$. The existence of a statistic satisfying (1) is guaranteed (with $m \leq 6$) by \Cref{lem:cyclic-stats-mod}; and picking the smallest valid $m$ lets us satisfy property (2). We use this to verify the claim by direct calculation, relying on the assumption that $s$ is symmetric mod $\ell$, meaning that $s_i = s_{i + \ell j}$ for all $i, j \in [k]$---interpreting all indices mod $k$: 
    \begin{align*}
        &\abs{S_{i_1,\ldots,i_m;\ell}(x-s) - S_{i_1,\ldots,i_m;\ell}(y-s)} \\
        =& \abs{\sum_{j=1}^{k/\ell} (x_{i_1 + \ell j} - s_{i_1 + \ell j})\ldots (x_{i_m + \ell j} - s_{i_m + \ell j}) - \sum_{j=1}^{k/\ell} (y_{i_1 + \ell j} - s_{i_1 + \ell j})\ldots (y_{i_m + \ell j} - s_{i_m + \ell j})} \\
        =& \abs{\sum_{j=1}^{k/\ell} \sum_{I \subseteq \{i_1,\ldots,i_m\}} \left( \prod_{i \in \{i_1,\ldots,i_m\} \backslash I} -s_{i + \ell j} \right) \left( \prod_{i \in I} x_{i + \ell j} - \prod_{i \in I} y_{i + \ell j} \right)} \\
        =& \abs{\sum_{I \subseteq \{i_1,\ldots,i_m\}} \left( \prod_{i \in \{i_1,\ldots,i_m\} \backslash I} -s_i \right) \sum_{j=1}^{k/\ell} \left( \prod_{i \in I} x_{i + \ell j} - \prod_{i \in I} y_{i + \ell j} \right)} \\
        =& \abs{\sum_{I \subseteq \{i_1,\ldots,i_m\}} \left( \prod_{i \in \{i_1,\ldots,i_m\} \backslash I} -s_i \right) \left( S_{I;\ell}(x) - S_{I;\ell}(y) \right)} \\
        =& \ \abs{S_{i_1,\ldots,i_m;\ell}(x) - S_{i_1,\ldots,i_m;\ell}(y)\tag{since for $I\subsetneq \{i_1,\ldots,i_m\}$ we have $S_{I;\ell}(x) = S_{I;\ell}(y)$}} \\
        \geq& \ 1
    \end{align*}
    where the last line follows from the facts that all cyclic statistics of integer sequences are integers and $S_{i_1,\ldots,i_m;\ell}(x) \neq S_{i_1,\ldots,i_m;\ell}(y)$. 
\end{proof}
\fi

\subsection{Our Algorithm}
\label{sec:algorithm}

We are now ready to describe our proposed tester, \Cref{alg:tester}. Given two candidates $x$ and $y$, a trace $\tz$ from the deletion channel can be described as a sequence of binomial random variables $\tz_1,\ldots,\tz_k$, each of which has an expectation in the set $\{qx_1,\ldots,qx_k,qy_1,\ldots,qy_k\}$. The algorithm begins by preprocessing the strings $x$ and $y$ to create a $C$-separated partition $c$ of the expectations (\Cref{line:partition}). We then consider the sequences $s_x$ and $s_y$ created by applying $c$ elementwise to $x$ and $y$. We design two tests for two separate cases:
\begin{itemize}
    \item If $s_x$ is a cyclic shift of $s_y$, then traces from $x$ and $y$ appear relatively similar. In this case, we invoke a subroutine (\Cref{alg:tester-similar}), which determines a shifted cyclic statistic in which $x$ and $y$ differ, then estimates this statistic from traces. Owing to the shift, the resulting estimate has low variance, allowing us to estimate the cyclic statistic from $\tilde{O}(n^6)$ traces.   
    \item If $s_x$ and $s_y$ are cyclically distinct, then the algorithm can reliably distinguish $x$ and $y$ by examining only a single trace whose $1$s are intact. Therefore, we draw a constant number of traces, choose an arbitrary trace $\tz$ with $k$ $1$s (\Cref{line:choice}), and cluster the gaps according to $c$. For $C$ large enough, the resulting sequence will with high probability correspond exactly to either $s_x$ or $s_y$.
\end{itemize}
Combining the results from each case yields our overall upper bound. We now state and prove the main result of \Cref{sec:upperbound}. Our proof relies on Lemmas~\ref{lem:centers-prob} and~\ref{lem:tester-similar} in \Cref{sec:similar-strings-ub}, which prove the correctness of \Cref{alg:tester-similar}. 

\subsubsection{Proof of Upper Bound}

\begin{algorithm}[tbp]
\caption{\textsc{Test-Cyclic-Traces}}\label{alg:tester}
\KwIn{$k \geq 0; x,y \in \mathbb{N}^{k}$; $q \in (0,1)$; sample access to deletion channel}
\KwOut{$x$ or $y$} 

\DontPrintSemicolon
\nl $C \gets $ absolute constant determined by \Cref{lem:centers-prob} \;
\nl $c \gets \textsc{Partition}(qx_1,\ldots,qx_k,qy_1,\ldots,qy_k; C)$ \label{line:partition} \;
\nl $s_x \gets (c(qx_1),\ldots,c(qx_k))$, $s_y \gets (c(qy_1),\ldots,c(qy_k))$ \;
\nl \If{$s_x$ is a cyclic shift of $s_y$} {
    \nl \Return \textsc{Test-Similar-Traces}($k$, $x$, $y$, $q$, $s_x$, $C$) \;
} \Else {
    \nl Draw $\Theta(1)$ traces from the deletion channel \;
    \nl $\tz \gets$ any trace with exactly $k$ $1$s \label{line:choice} \;
    \nl \textbf{if} $c(\tz_1),\ldots,c(\tz_k)$ is a cyclic shift of $s_x$ \textbf{then} \Return $x$ \textbf{else} \Return $y$ \;
}
\end{algorithm}

\begin{theorem}
    \label{thm:upperbd}
    Let $x$ and $y$ be two cyclically distinct $k$-sparse binary strings of length $n$. Then Algorithm~\ref{alg:tester} distinguishes $x$ from $y$ with probability $\geq\frac{2}{3}$ for sufficiently large $n$, using $O(n^6\log^{12}n)$ traces from a cyclic deletion channel.  
\end{theorem}
\begin{proof}
    We analyze \Cref{alg:tester} to show the theorem. We consider two cases:
    \begin{description}
        \item[Case 1:] $s_x$ is a cyclic shift of $s_y$. In this case, we invoke \Cref{alg:tester-similar}, which draws $O(n^6\log^{12}n)$ returns the correct answer with probability at least $2/3$, by \Cref{lem:tester-similar}.

        \item[Case 2:] $s_x$ is not a cyclic shift of $s_y$. Then we draw $T = \log(1/4)/\log (1-q^k) = O(1)$ traces. With probability $3/4$, at least one trace $\tz$ will contain exactly $k$ $1$s.  Suppose $\tz \sim \Del{x}$. Due to Lemma~\ref{lem:centers-prob}, with probability at least $1-O(n^{-10})$, we have $c(\tz_j) = c(qx_j)$ for all $j \in [k]$. Therefore, $c(\tz) = s_x$, and the algorithm returns $x$ with probability at least $3/4 - O(n^{-10}) \geq 2/3$ for large enough $n$.  
    \end{description}
    In both cases, the number of traces $T$ is bounded by $O(n^6\log^{12}n)$.
\end{proof}

\subsubsection{Distinguishing Similar Strings}
\label{sec:similar-strings-ub}

We now handle the case in which $s_x = s_y$, i.e., $x$ and $y$ are cyclically equivalent under the $C$-separated partition $c$ constructed in Algorithm~\ref{alg:tester}. We build an estimator based on shifted cyclic statistics and show that $\tilde{O}(n^6)$ traces suffice to distinguish $x$ from $y$. In particular, let $\ell$ be the minimum value such that $s_x$ and $s_y$ are symmetric mod $\ell$. For each $j \in [k]$, let $P_j \coloneqq \{x_{j'} : c(qx_{j'}) = j\} \cup \{y_{j'} : c(qy_{j'}) = j\}$. Define the function $g: [k] \rightarrow \mathbb{R}$ to be the average of the points in cluster $j$; that is, $g(j) = \frac{1}{\abs{P_j}} \sum_{z \in P_j} qz$. The algorithm chooses the smallest value of $m$ for which there is an order $m$ cyclic statistic $S_{i_1,\ldots,i_m;\ell}$ such that $S_{i_1,\ldots,i_m;\ell}(x) \neq S_{i_1,\ldots,i_m;\ell}(y)$. We call a trace $\tz$ {\bf useful} if it  satisfies the following conditions:
    \begin{enumerate}
        \item There are exactly $k$ $1$s in $\tz$.
        \item For all $j \in [k]$, $\abs{\tz_j - g(c(\tz_j))} \leq (4k+1)C\sqrt{n}\log n$. 
    \end{enumerate}
Intuitively, the second condition says that the gaps in $\tz$ can be matched up with cluster centers as we expect. Suppose there are $T'$ useful traces. Our estimator will apply the statistic $S_{i_1,\ldots,i_m;\ell}(x)$  to all the useful traces after subtracting out their cluster centers:
\begin{equation}
        \label{eq:estimator}
        f_{i_1,\ldots,i_m;\ell}(\tz) = \frac{1}{q^m}\sum_{j=1}^{k/\ell} \left(\tz_{i_1+j\ell} - g\left(c(\tz_{i_1+j\ell})\right)\right) \ldots \left(\tz_{i_m+j\ell} - g\left(c(\tz_{i_m+j\ell})\right)\right)
\end{equation}
We define $\hat{f} = \frac{1}{T'} \sum_{\tz:\textrm{useful}(\tz)} f_{i_1,\ldots,i_m;\ell}(\tz)$ as the average of the above equation over useful traces. We show that, for any set of $\tilde{\Theta}(n^6)$ traces from a string $z$, $\abs{\hat{f} - S_{i_1,\ldots,i_m;\ell}(z)} \leq 1/3$ with probability at least $2/3$, which implies an efficient tester. 

The main result of this subsection requires the following lemma, which roughly states that we can recover clusters from a random trace with high probability. 

\begin{lemma}
    \label{lem:centers-prob}
    Let $z_1,\ldots,z_{2k} \in \mathbb{N}$, $q \in (0,1)$, and $\tz_j \sim \Bin{z_j}{q}$ for all $j \in [k]$. Then there exists an absolute constant $C > 0$ such that, for any $C$-separated partition $c$ of $qz_1,\ldots,qz_{2k}$, $|\tz_j - qz_j| \leq C\sqrt{n}\log n$ and $c(\tz_j) = c(qz_j)$ for all $j \in [k]$ with probability at least $1 - O_k(n^{-10})$. 
\end{lemma}
\begin{proof}
    Fix $j \in [k]$. Let $c$ be a $C$-separated partition of $qz_1,\ldots,qz_{2k}$. By a Chernoff bound, there exists a constant $C$ such that $|\tz_j - qz_j| \leq C\sqrt{n}\log n$ with probability at least $1 - n^{-10}$. When this occurs, the following holds for all $j' \in [k]$ with $c(qz_{j'}) \neq c(qz_j)$:
    \begin{equation*}
        |\tz_j - qz_{j'}| \geq |qz_j - qz_{j'}| - |\tz_j - qz_{j}| > 2C\sqrt{n}\log n - C \sqrt{n}\log n= C\sqrt{n}\log n \geq |\tz_j - qz_{j}|
    \end{equation*}
    In this case, clearly $c(\tz_j) = c(qz_j)$. The contrapositive of this is that, if $c(\tz_j) \neq c(qz_j)$, we must have $|\tz_j - qz_j| > C\sqrt{n}\log n$. Taking a union bound over the $k$ random variables, the probability that this occurs for some $j \in [k]$ is at most $kn^{-10}$, as desired. 
\end{proof}

\begin{algorithm}[tbp]
\caption{\textsc{Test-Similar-Traces}}\label{alg:tester-similar}
\KwIn{$k \geq 0; x,y \in \mathbb{N}^{k}$; $q \in (0,1)$; $s \in \mathbb{R}^k$; $C > 0$; sample access to deletion channel}
\KwOut{$x$ or $y$} 

\DontPrintSemicolon
\nl $\ell \gets$ minimum value such that $s$ is symmetric mod $\ell$ \;
\nl \label{line:cyclic-stat}Let $S_{i_1,\ldots,i_m;\ell}$ be the statistic from Lemma~\ref{lem:stats-diff} \;
\nl Draw $T = \Theta(n^6\log^{12} n)$ traces $\tz^{(1)},\ldots,\tz^{(T)}$ from the deletion channel \;
\nl Let $\hat{f}$ be the average of $f_{i_1,\ldots,i_m;\ell}(\tz^{(t)})$ (from \Cref{eq:estimator}) for all useful traces $\tz^{(t)}$\;
\nl \Return $\argmin_{z \in \{x,y\}} \abs{\hat{f} - S_{i_1,\ldots,i_m;\ell}(z - \frac{1}{q}g(s))}$ \label{line:decision} \;
\end{algorithm}

We are now ready to prove the correctness of \Cref{alg:tester-similar}.

\begin{lemma}
    \label{lem:tester-similar}
    Let $x = (x_1,\ldots,x_k)$ and $y = (y_1,\ldots,y_k)$ be two cyclically distinct binary strings of length $n$, and $q \in (0,1)$. Let $c$ be a $C$-separated partition of $qx_1,\ldots,qx_k,qy_1,\ldots,qy_k$, and suppose $c(qx_1),\ldots,c(qx_k)$ is a cyclic shift of $c(qy_1),\ldots,c(qy_k)$. Then for sufficiently large $C$ and $n$, there exists an algorithm (\Cref{alg:tester-similar}) which distinguishes $x$ from $y$ with probability at least $2/3$ using $T = O(n^6\log^{12}n)$ traces from a cyclic deletion channel.
\end{lemma}
\begin{proof}
    Fix $C$ to be the constant determined by \Cref{lem:centers-prob}. Let $\mathcal{Z}$ denote the set of useful traces (defined at the top of \Cref{sec:similar-strings-ub}).
    The first usefulness condition holds with probability exactly $q^k$ for each trace. By \Cref{lem:centers-prob}, we have $\abs{\tz_j - qz_j} \leq C\sqrt{n}\log n$ and $c(\tz_j) = c(qz_j)$ with probability at least $1 - O(n^{-10})$. In this case, we have $\abs{\tz_j - g(c(\tz_j))} \leq \abs{\tz_j - qz_j} + \abs{qz_j - g(c(\tz_j))} \leq  C\sqrt{n}\log n + 4Ck\sqrt{n}\log n$, which is precisely the second condition. Therefore, for large enough $n$, each trace falls in $\mathcal{Z}$ with constant probability. By a Chernoff bound, $\abs{\mathcal{Z}} = \Omega(T)$ with probability $1 - o(1)$. For each $\tz \in \mathcal{Z}$, the algorithm computes $f_{i_1,\ldots,i_m;\ell}(\tz)$, where $f_{i_1,\ldots,i_m;\ell}$ is defined in \Cref{eq:estimator}. 
    Let $A$ be the event that $\tz \in \mathcal{Z}$ and $c(\tz_j) = c(qz_j)$ for all $j \in [k]$. Conditioned on $A$, it is possible to compute $f(\tz)$ precisely, since we can choose a consistent cyclic shift of $s$ and match $\tz$ to that shift up to symmetry. We therefore have 
    \begin{align*}
        \E{f(\tz) \mid A} &= \E{\frac{1}{q^m}\sum_{j=1}^{k/\ell} \left( \tz_{i_1+j\ell} - g(c(qz_{i_1+j\ell})) \right) \ldots \left( \tz_{i_m+j\ell} - g(c(qz_{i_m+j\ell})) \right)}\pm O(n^{-4}) \\
        &= \frac{1}{q^m}\sum_{j=1}^{k/\ell} \E{\tz_{i_1+j\ell} - g(c(qz_{i_1+j\ell}))} \ldots \E{\tz_{i_m+j\ell} - g(c(qz_{i_m+j\ell}))}\pm O(n^{-4}) \tag{$\tz_j$s are independent} \\
        &= \frac{1}{q^m}\sum_{j=1}^{k/\ell} (qz_{i_1+j\ell} - g(s_{i_1+j\ell})) \ldots (qz_{i_m+j\ell} - g(s_{i_m+j\ell}))\pm O(n^{-4})  \\
        &= S_{i_1,\ldots,i_m;\ell}(z-\frac{1}{q}g(s)) \pm O(n^{-4}) 
    \end{align*}
    where the $\pm O(n^{-4})$ term in the first line comes from the $O(n^{-10})$ probability that there are $k$ 1s in $\tilde{z}$ but event $A$ fails; this is multiplied by the universal $O(n^6)$ bound on $|f(\tilde{z})|$ for $m\leq 6$.
    
    By \Cref{lem:centers-prob}, $A$ holds with probability at least $1 - O(n^{-10})$ conditioned on $\tz \in \mathcal{Z}$. Additionally, the second condition implies that for all $\tz \in \mathcal{Z}$, we have $\abs{f(\tz)} \leq \frac{k}{q^m}(4k+1)^6C^6n^3\log^6 n$. Combining these facts and applying the law of total probability yields
    \begin{equation*}
        \abs{\E{f(\tz) \mid \tz \in \mathcal{Z}} - S_{i_1,\ldots,i_m;\ell}(z-\frac{1}{q}g(s))} \leq O \left( \frac{n^3\log^6 n}{n^{10} } + n^{-4}\right) = o(1) 
    \end{equation*}
    Combining the bound on $\abs{f(\tz)}$ with Hoeffding's inequality, we have
    \begin{align*}
        \prob{\abs{\hat{f} - \E{f(\tz \mid z \in \mathcal{Z})}} \geq \frac{1}{4}} &\leq 2 \exp \left( -\frac{2\abs{\mathcal{Z}}^2}{16\abs{\mathcal{Z}} \cdot \left(\frac{k}{q^m}(4k+1)^6C^6 n^3 \log^6n\right)^2} \right) \\
        &= 2 \exp \left( -\frac{\abs{\mathcal{Z}}}{16\frac{k^2}{q^{2m}}(4k+1)^{12}C^{12}n^{6}\log^{12}n} \right) 
    \end{align*}
    By our choice of $T$, we have $\abs{\mathcal{Z}} = \Omega(n^{6}\log^{12}n)$ with probability $1 - o(1)$. Therefore, for an appropriate choice of constants, the above equation is bounded by $0.1 + o(1)$. By the triangle inequality, the following holds for large enough $n$ with probability at least $0.8$: $\abs{\hat{f} - S_{i_1,\ldots,i_m;\ell}(z-\frac{1}{q}g(s))} \leq \frac{1}{4} + o(1) \leq \frac{1}{3}$. By \Cref{lem:stats-diff}, there must exist a statistic $S_{i_1,\ldots,i_m;\ell}$ such that $\abs{S_{i_1,\ldots,i_m;\ell}(x-\frac{1}{q}g(s)) - S_{i_1,\ldots,i_m;\ell}(y-\frac{1}{q}g(s))} \geq 1$ (namely, the statistic chosen on \Cref{line:cyclic-stat}). Therefore, if $z = x$, the algorithm returns $x$ with probability at least $0.9$ (\Cref{line:decision}). The same is true for $y$, which completes the proof.
\end{proof}

\section{Lower Bounds}\label{sec:lower}
Our goal in this section is to upper-bound the distance between the distributions $\Del{x}$ and $\Del{y}$ for two strings $x,y$ with that are permutations of each other and have identical low-order cyclic statistics, which will yield a lower bound for the number of traces needed to distinguish these strings. We first analyze the probability of observing a given trace $a$ from $\Del{x}$, viewed as a polynomial in $n$. We show that the higher-order coefficients of this polynomial depend on the lower-order cyclic statistics of $x$. Using this fact, we bound the distance between probabilities from $\Del{x}$ versus $\Del{y}$. We conclude by providing two strings $x,y$ with identical cyclic statistics up to order 4, proving our $\tilde{\Omega}(n^5)$ lower bound.

\begin{lemma}
    \label{lem:probability}
    Let $x = 10^{n+x_1}10^{n+x_2}\ldots10^{n+x_k}$ with $x^*$ an integer upper bound on $x_1,\ldots,x_k$, and let $a$ be any cyclic shift of $10^{a_1}10^{a_2}\ldots10^{a_k}$. Then for a circular deletion channel $\Del{x}$ with deletion probability $p$, letting $n'=pn$, and letting $b_j=a_j-n(1-p)$ for $j\in [k]$, we have, for $Sym(x;n,p,b,x^*)$ some symmetric function of $x_1,\ldots,x_k$, that depends on $n,p,b,x^*$, that
    \begin{equation}\label{eq:trace-probability-lemma}
        \prob[\tx \sim \Del{x}]{\tx = a} = Sym(x;n,p,b,x^*) \sum_{i=1}^k \prod_{j = 1}^k \prod_{h = x_j+1}^{x^*}(n'-b_{j+i}+h),
    \end{equation}
    where all indices are taken modulo $k$.
\end{lemma}
\begin{proof}
    Since all cyclic shifts of a trace have equal probability under the deletion channel, we assume without loss of generality that $a = 10^{a_1}10^{a_2}\ldots10^{a_k}$. By definition of a circular deletion channel, and letting $|a|,|x|$ denote the length of the underlying binary strings $a,x$ respectively, we have
    \begin{equation}
        \label{eq:trace-probability}
        \prob[\tx \sim \Del{x}]{\tx = a} = \frac{1}{|a|}\left(\sum_{i=1}^k \prod_{j = 1}^k \binom{n+x_j}{a_{j+i}}\right)p^{|x| - |a|}q^{|a|}
    \end{equation}
    where indices are taken modulo $k$. We point out that the terms outside the parentheses, $\frac{1}{|a|}p^{|x| - |a|}q^{|a|}$ are a symmetric function of $x$ (depending only on $|x|$), and thus can be absorbed into the initial $Sym(x;n,p,b,x^*)$ term in the lemma statement. We slightly refactor each binomial term as
    \[\binom{n+x_{j}}{a_{i+j}}=\frac{n!}{a_{i+j}!(n+x^*-a_{i+j})!} \left(\prod_{h=1}^{x_{j}}(n+h) \right)\left(\prod_{h=x_{j}+1}^{x^*}(n-a_{i+j}+h)\right)\]

    We thus need only consider the sum over $i$ of the product over $j$ of the above expression. The first term $\frac{n!}{a_{i+j}!(n+x^*-a_{i+j})!}$ does not depend on $x$ so is by definition symmetric in $x$ and can be absorbed into the initial $Sym(x;n,p,b,x^*)$. The next term, $\prod_{h=1}^{x_{j}}(n+h)$, when we take its product over all $j$, becomes a symmetric function of $x_1,\ldots,x_k$ and can thus also be absorbed into $Sym(x;n,p,b,x^*)$. The final term becomes the last term in Equation~\ref{eq:trace-probability-lemma}, after noting that $n-a_{i+j}=n'-b_{i+j}$, by definition of $n',b$.
\end{proof}

The only non-symmetric portion of Equation~\ref{eq:trace-probability-lemma} is the expression inside the sum, which is 
\begin{equation}\label{eq:probability-term}\sum_{i=1}^k \prod_{j = 1}^k \prod_{h = x_j+1}^{x^*}(n'-b_{j+i}+h).\end{equation}
Thus, when we compare traces from $\Del{x}$ versus $\Del{y}$ for sequences $x,y$ that are permutations of each other, it is this last term that will distinguish between them. We now view this term as a polynomial in $n'$, and show that the high-order coefficients of this polynomial can be expressed in terms of low-order cyclic statistics of $x$. Thus for two sequences $x,y$ with identical low-order cyclic statistics, we expect the distributions of $\Del{x}$ versus $\Del{y}$ to have high-order terms in $n'$ that exactly cancel, leading to our main lower bound.

\begin{lemma}\label{lem:cyclic-form}Expression \ref{eq:probability-term} when viewed as a polynomial in $n'$ of degree $L=\sum_{j=1}^k x^*-x_j$, for any $m\geq 0$ has ${n'}^{L-m}$ coefficient that is a linear combination of degree $\leq m$ cyclic statistics of $x$, each multiplied by some symmetric function of $x$ (possibly depending on $n,p,b,x^*$).
\end{lemma}
\begin{proof}
Define $S$ to be the set of pairs $(j,h)$ such that $j\in[k]$ and $h\in \{x_j+1,\ldots,x^*\}$, where $|S|=L$ as defined in the lemma. Thus the main expression of the lemma equals $\sum_{i=1}^k \prod_{(j,h)\in S} (n'-b_{j+i}+h)$. In particular, the coefficient of ${n'}^{L-m}$ in this expression can be found by multiplying the non-$n'$ parts of all combinations of $m$ \emph{distinct} terms: \begin{equation}\label{eq:distinct} \sum_{i=1}^k\sum_{\substack{\text{distinct }(j_1,h_1),\\\ldots,(j_m,h_m) \in S}} (h_1 - b_{j_1+i})\ldots(h_m - b_{j_m+i})\end{equation}
For each fixed $i$, this expression is a \emph{symmetric} polynomial of the terms $(h-b_{i+j})$ for $(j,h)\in S$, and we thus use Newton's symmetric polynomial identities to conclude that we may reexpress this expression for fixed $i$ as a linear combination of products of ``power sums'' of these terms; the power sum of degree $\ell$ is defined as $\sum_{(j,h) \in S} (h - b_{j+i})^\ell$. Namely, Equation~\ref{eq:distinct} can be expressed as some linear combination of the following expressions, where $\ell_1\ldots,\ell_s$ are some positive integers that sum to $m$:
\begin{equation}\label{eq:h-b} \sum_{i=1}^k\prod_{r=1}^s \sum_{(j,h)\in S} (h-b_{i+j})^{\ell_r}\end{equation}
We break down the inner sum of this expression using the definition of $S$, and the binomial expansion (across different powers $t$), as 
\[\sum_{(j,h)\in S} (h-b_{i+j})^{\ell_r}=\sum_{j=1}^k \sum_{t=0}^{\ell_r} {\binom{\ell_r}{t}} (-b_{i+j})^{\ell_r-t}\sum_{h=x_j}^{x^*} h^t.\]
For any exponent $t$, we consider the final sum as a degree $t+1$ polynomial in $x_j$, with notation $P_{t+1}(x_j):=\sum_{h=x_j}^{x^*} h^t$. 
Thus Equation~\ref{eq:h-b} equals
\[\sum_{i=1}^k\prod_{r=1}^s\left(\sum_{j_r=1}^k \sum_{t_r=0}^{\ell_r} {\ell_r\choose t_r} (-b_{i+j_r})^{\ell_r-t_r} P_{t_r+1}(x_{j_r})\right).\]
We pull the sums over $t_r$ outside:
\[ \sum_{t_1=0}^{\ell_1}\cdots
\sum_{t_s=0}^{\ell_s}\sum_{i=1}^k\prod_{r=1}^s\left(\sum_{j_r=1}^k {\ell_r\choose t_r} (-b_{i+j_r})^{\ell_r-t_r} P_{t_r+1}(x_{j_r})\right)\]
We split the product by whether, for each $r$, we have $t_r=\ell_r$, since when this is true the expression simplifies significantly (note that we also move the sum over $i$ to the right, past the first parenthetical, which does not depend on $i$):
\[\sum_{t_1=0}^{\ell_1}\cdots
\sum_{t_s=0}^{\ell_s}\left(\prod_{r\in [s]:t_r=\ell_r}  \sum_{j_r=1}^k P_{\ell_r+1}(x_{j_r})\right)\left(\sum_{i=1}^k\prod_{r\in [s]:t_r<\ell_r} \sum_{j_r=1}^k {\ell_r\choose t_r} (-b_{i+j_r})^{\ell_r-t_r} P_{t_r+1}(x_{j_r})\right)\]
The first parenthetical is clearly a symmetric function of $x_1,\ldots,x_k$. To analyze the second expression, we apply a variable substitution in the inner sum, replacing $j_r$ with $j_r-i$. For any fixed tuple $t_1,\ldots,t_s$, the last parenthetical equals 
\[\sum_{i=1}^k\prod_{r\in [s]:t_r<\ell_r} \sum_{j_r=1}^k {\ell_r\choose t_r} (-b_{j_r})^{\ell_r-t_r} P_{t_r+1}(x_{j_r-i}).\]

We view the product as a polynomial in the variables $x_{1-i},\ldots,x_{k-i}$, which we denote as $Q(x_{1-i},\ldots,x_{k-i})$; the degree of this polynomial is thus bounded by  $\sum_{r:t_r<\ell_r}(t_r+1)$, which is at most $\sum_r \ell_r=m$.
Crucially, the sum over $i$ now equals $\sum_{i=1}^k Q(x_{1-i},\ldots,x_{k-i})$, which, since we sum a degree $m$ polynomial over all cyclic shifts of $x_1,\ldots,x_k$, is thus clearly a linear combination of cyclic statistics of degree $\leq m$.

In conclusion, Equation~\ref{eq:h-b} is thus a sum of cyclic statistics of $x$ each time some symmetric function of $x$. Thus, the coefficient of ${n'}^{L-m}$ in Equation~\ref{eq:probability-term}, being a linear combination of expressions of the form of Equation~\ref{eq:h-b}, is also a sum of cyclic statistics of $x$ each time some symmetric function of $x$, as desired. 
\end{proof}

We now directly compare the probabilities of getting a certain trace $a$ from $\Del{x}$ versus $\Del{y}$.

\begin{lemma}
    \label{lem:ratio}
    Let $x = 10^{n+x_1}10^{n+x_2}\ldots10^{n+x_k}$ and $y = 10^{n+y_1}10^{n+y_2}\ldots10^{n+y_k}$. Suppose $(x_1,\ldots,x_k)$ is a permutation of $(y_1,\ldots,y_k)$ and that the two sequences have matching cyclic statistics up to some order $z\geq 1$. Let $a = 10^{a_1}10^{a_2}\cdots10^{a_k}$. If $a_j-n(1-p) \in \left[-C\sqrt{n\log n}, C\sqrt{n\log n}\right]$ for all $j \in [k]$ then for sufficiently large $n,c$ depending on $C$, $k$, $p$, and $x^*=\max\{x_1,\ldots,x_k\}$ we have
    \begin{equation*}
        \frac{\prob[\tx \sim \Del{x}]{\tx = a}}{\prob[\tilde{y} \sim \Del{y}]{\tilde{y} = a}} \in \left[1-(c\log n/n)^{\frac{z+1}{2}}, 1+(c\log n/n)^{\frac{z+1}{2}}\right].
    \end{equation*}
\end{lemma}
\begin{proof}
As above, define $b_j=a_j-n(1-p)$ and define $n'=pn$. Since $x^*$ upper bounds $x_1,\ldots,x_k$, we have that $x^*$ upper bounds $y_1,\ldots,y_k$ too, since they are a permutation of $x_1,\ldots,x_k$.
    By Lemma~\ref{lem:probability}, we have
    \[\prob[\tx \sim \Del{x}]{\tx = a} = Sym(x;n,p,b,x^*) \sum_{i=1}^k \prod_{j = 1}^k \prod_{h = x_j+1}^{x^*}(n'-b_{j+i}+h)\]
    and the corresponding expression for $y$.
    Let $D_1 \coloneqq \sum_{i=1}^k \prod_{j = 1}^k \prod_{h = x_j+1}^{x^*}(n'-b_{j+i}+h)$ and $D_2 \coloneqq \sum_{i=1}^k \prod_{j = 1}^k \prod_{h = y_j+1}^{x^*}(n'-b_{j+i}+h)$. Because $(y_1,\ldots,y_k)$ is a permutation of $(x_1,\ldots,x_k)$, we have 
    \begin{equation*}
        \frac{\prob[\tx \sim \Del{x}]{\tx = a}}{\prob[\ty \sim \Del{y}]{\ty = a}} = \frac{\sum_{i=1}^k \prod_{j = 1}^k \prod_{h = x_j+1}^{x^*}(n'-b_{j+i}+h)}{\sum_{i=1}^k \prod_{j = 1}^k \prod_{h = y_j+1}^{x^*}(n'-b_{j+i}+h)} = \frac{D_1}{D_2} = 1+\frac{D_1-D_2}{D_2}
    \end{equation*}
    It remains only to show that $\frac{D_1-D_2}{D_2} \in \left[-(c\log n/n)^{\frac{z+1}{2}}, (c\log n/n)^{\frac{z+1}{2}}\right]$. 
    
    Let $L \coloneqq kx^* - \sum_{j=1}^k x_j$. Then $D_1,D_2$ are both degree $L$ polynomials in $n'$ with leading coefficient $k$. For $m\geq 1$ we now bound the contribution of the terms of degree $\leq L-m$. 

    By assumption, each $b_j\in \left[-C\sqrt{n\log n}, C\sqrt{n\log n}\right]$; thus for $C'=C+x^*$ and $n\geq 2$ we trivially have that $|-b_j+h|\leq C'\sqrt{n\log n}$, for $h\in \{1,\ldots,x^*\}$. Thus by definition of $D_1,D_2$, the contribution of the ${n'}^{L-m}$ term to either $D_1$ or $D_2$ has magnitude at most $k{L\choose m} (C'\sqrt{n\log n})^m {n'}^{L-m}$. We bound ${L\choose m}\leq L^m$, and also bound $L\leq k x^*$. We use these bounds to bound the contribution to $D_1$ or $D_2$ of \emph{all} terms of degree $\leq L-m$ in $n'$: this is bounded by the geometric series $k\sum_{i\geq m} (kx^* C' \sqrt{n\log n})^i {n'}^{L-i}$. We choose $n$ large enough so that each term of this series is at most $\frac{1}{3}$ of the previous term, namely, we choose $n$ so that $\frac{\sqrt{n}}{\log n}\geq 3\frac{kx^* C'}{p}$. Because the ratio of terms is $\leq\frac{1}{3}$, the $i=0$ term of the series is at least twice as large as the sum of magnitudes the remaining terms combined, so that, since the $i=0$ term of $D_2$ is $k (n')^L$, we conclude that  $D_2\geq \frac{1}{2} k {n'}^L$. Second, the sum of the series starting with the $i=m$ term has magnitude at most $\frac{3}{2}$ times the $i=m$ term, namely, \[k\sum_{i\geq m} (kx^* C' \sqrt{n\log n})^i {n'}^{L-i}\leq \frac{3}{2}k (kx^* C' \sqrt{n\log n})^m {n'}^{L-m}=\frac{3}{2}k{n'}^L \left(\frac{(kx^*C')^2\log n}{p^2 n}\right)^{\frac{m}{2}}\]
    
    We now invoke Lemma~\ref{lem:cyclic-form}, which says that the ${n'}^{L-m}$ coefficients of $D_1,D_2$ respectively are linear combinations of degree $\leq m$ cyclic statistics of $x,y$ respectively, times symmetric functions of $x,y$ respectively. The assumption that $x,y$ are permutations of each other means that all symmetric functions of $x$ equal the corresponding symmetric functions of $y$; the assumption that $x,y$ have identical cyclic statistics up to order $z$ further implies that the ${n'}^{L-m}$ coefficients of $D_1,D_2$ must be identical for $m\leq z$. Thus we bound $|D_1-D_2|$ by plugging in the bounds of the previous paragraph starting at the first nonzero degree, $m=z+1$: $|D_1-D_2|\leq 2\cdot\frac{3}{2}k{n'}^L \left(\frac{(kx^*C)^2\log n}{p^2 n}\right)^{\frac{z+1}{2}}$. Combining with our lower bounds from above that $D_2\geq \frac{1}{2} k {n'}^L$, and setting $c=6\frac{(kx^*C)^2}{p^2}$, we conclude $\frac{D_1-D_2}{D_2} \in \left[-(c\log n/n)^{\frac{z+1}{2}}, (c\log n/n)^{\frac{z+1}{2}}\right]$ as desired.
\end{proof}
We use the following bound relating the Hellinger distance between $\mu,\nu$ to the TV distance between $t$ samples from $\mu$ or $\nu$.
\begin{lemma}[Lemma A.5 in~\cite{holden20Lower}]
    \label{lem:hellinger}
    Let $\mu$ and $\nu$ be probability measures with $d_H(\mu, \nu) \leq 1/2$. Then $1 - \DTV(\mu^t,\nu^t) \leq \eps$ if $t \leq \frac{\log(1/\eps)}{9d_H(\mu,\nu)}$. 
\end{lemma}
\begin{lemma}
    \label{lem:lowerbd}
    Let $x = 10^{n+x_1}10^{n+x_2}...10^{n+x_k}$ and $y = 10^{n+y_1}10^{n+y_2}...10^{n+y_k}$ be two strings where $(x_1,...x_k)$ and $(y_1,...,y_k)$ are permutations of each other and have identical cyclic statistics up to order $z$. Then any algorithm which distinguishes $\Del{x}$ from $\Del{y}$ with probability $\frac{2}{3}$ requires $\Omega(n^{z+1}/\log^{z+1}n)$ samples. 
\end{lemma}
\begin{proof}
    Let $\mu_0$ represent $\Del{x}$ conditioned on deleting only $0$s; with $\nu_0$ defined correspondingly for $\Del{y}$. The distributions $\mu_0,\nu_0$ are only supported on traces that are cyclic shifts of $a = 10^{a_1}10^{a_2}\cdots10^{a_k}$, for some gaps $a_1,\ldots,a_k$.
     
     As above, define $b_i=a_i-n(1-p)$.     
     We point out that, by a standard Chernoff bound, we have $b_i \in [-C\sqrt{n\log n}, C\sqrt{n\log n}]$ with probability at least $1-n^{-(z+100)}$ for some universal constant $C$ (assuming $z$ is a constant).

    Combining Lemma \ref{lem:ratio} with this Chernoff bound yields the following bound on the Hellinger distance between $\mu_0$ and $\nu_0$:
    \begin{align*}
        &d_H^2(\mu_0, \nu_0) = \sum_{a_1,...,a_k} \left( \sqrt{\mu_0(10^{a_1}...10^{a_k})} - \sqrt{\nu_0(10^{a_1}...10^{a_k})} \right)^2 \\
        &= \frac{1}{(1-p)^k} \sum_{a_1,...,a_k} \left( \sqrt{\Pr_{a'\sim \Del{x}}[a'=a]} - \sqrt{\Pr_{a'\sim \Del{y}}[a'=a]} \right)^2 \\
        &\leq \frac{2 kn^{-(z+100)}}{(1-p)^k} + \frac{1}{(1-p)^k} \hspace{-1.6cm}\sum_{\substack{a_1,...,a_k \in\\ [n(1-p)-C\sqrt{n\log n}, n(1-p)+C\sqrt{n\log n}]}} \hspace{-1.6cm}\Pr_{a'\sim \Del{y}}[a'=a] \left( 1 - \sqrt{\frac{\Pr_{a'\sim \Del{x}}[a'=a]}{\Pr_{a'\sim \Del{y}}[a'=a]}} \right)^2 \\
        &= O\left( \frac{\log^{z+1} n}{n^{z+1}} \right)
    \end{align*}
    We can generate a sample from $\Del{x}$ by first sampling from $\mu_0$ and then passing the sample through a second channel which deletes only $1$s (and similarly for $\Del{y}$ and $\nu_0$). By the data processing inequality, the second channel cannot increase the Hellinger distance between $\mu_0$ and $\nu_0$, so we have $d_H(\Del{x}, \Del{y}) \leq d_H(\mu_0, \nu_0) = O\left( \frac{\log^{z+1} n}{n^{z+1}} \right)$. Applying Lemma \ref{lem:hellinger}, we conclude that one requires $\Omega \left( \frac{n^{z+1}}{\log^{z+1} n} \right)$ traces to distinguish $x$ from $y$ with constant success probability. 
\end{proof}

Finally, we provide two strings that have identical cyclic statistics up to $4$-th order, and whose gaps are permutations of each other. They are the sequence $x_{(j)}=(0,2,3,2,1,1,1,1,2,3,2,0)$ and the sequence $y_{(j)}=3-x_{(j)}$.
Thus, distinguishing these two strings requires $\Omega(\frac{n^5}{\log^5 n})$ traces.

\bibliography{references}
\bibliographystyle{plainnat}

\end{document}